\documentclass[12pt, svgnames]{article}
\usepackage{amsmath}
\usepackage{amssymb}
\usepackage{amsthm}
\usepackage{mathrsfs}
\usepackage{bbm}
\usepackage{tikz}
\usetikzlibrary{decorations.markings,decorations.pathreplacing,calc}
\usetikzlibrary{arrows.meta}
\usepackage{subcaption}

\usepackage[framemethod=tikz]{mdframed}
\usepackage[margin=1in]{geometry}
\usepackage[utf8]{inputenc}
\usepackage{comment}
\usepackage{umoline}

\usepackage{hyperref}
\hypersetup{
     colorlinks=true,           
     linkcolor=purple,          
     citecolor=blue,            
     filecolor=blue,            
     urlcolor=cyan,             
 }

\theoremstyle{plain}
\newtheorem{theorem}{Theorem}[section]
\newtheorem{lemma}[theorem]{Lemma}
\newtheorem{corollary}[theorem]{Corollary}
\theoremstyle{definition}

\newtheorem{claim}[theorem]{Claim}
\newtheorem{fact}[theorem]{Fact}

\newcommand {\br} [1] {\ensuremath{ \left( #1 \right) }}
\newcommand {\Br} [1] {\ensuremath{ \left[ #1 \right] }}

\newcommand {\norm}[1]{{\| #1 \|}}  
\newcommand {\bra} [1] {\ensuremath{ \left\langle #1 \right| }}
\newcommand {\ket} [1] {\ensuremath{ \left| #1 \right\rangle }}
\newcommand {\ketbratwo} [2] {\ensuremath{ \left| #1 \middle\rangle \middle\langle #2 \right| }}
\newcommand {\ketbra} [1] {\ketbratwo{#1}{#1}}
\DeclareMathOperator*{\Tr}{Tr}
\DeclareMathOperator*{\DL}{DL}
\newcommand {\id} {\ensuremath{\mathbbm{1}}}

\newcommand {\suppress}[1]{}

\def\max{\mathrm{max}}
\def\min{\mathrm{min}}

\newcommand{\Loc}{\mathrm{Loc}}
\newcommand{\Ind}{\mathrm{Ind}}
\newcommand{\hQ}{\widehat{Q}}

\newcommand{\op}{\mathrm{Op}}
\newcommand{\SR}{\mathrm{SR}}

\newcommand{\sign}{\mathrm{sign}}
\newcommand{\inte}{\mathrm{int}}

\newcommand{\EqDef}{\stackrel{\mathrm{def}}{=}}

\def\cT{\mathcal{T}}

\def\cG{\mathcal{G}}

\newcommand{\Eq}[1]{Eq.~(\ref{#1})}
\newcommand{\Ref}[1]{Ref.~\cite{#1}}
\newcommand{\Fig}[1]{Fig.~\ref{#1}}
\newcommand{\Sec}[1]{Sec.~\ref{#1}}
\newcommand{\Thm}[1]{Theorem~\ref{#1}}
\newcommand{\Lem}[1]{Lemma~\ref{#1}}
\newcommand{\App}[1]{Appendix~\ref{#1}}

\newcommand{\footremember}[2]{%
    \footnote{#2}
    \newcounter{#1}
    \setcounter{#1}{\value{footnote}}%
}
\newcommand{\footrecall}[1]{%
    \footnotemark[\value{#1}]%

} \title{Entanglement subvolume law for 2D frustration-free spin systems}
\author{%
  Anurag Anshu \footremember{iqc}{Institute for Quantum Computing, University of Waterloo, Canada} \footremember{co}{Department of Combinatorics and Optimization, University of Waterloo, Canada}\footremember{pi}{Perimeter Institute for Theoretical Physics, Canada}%
  \and Itai Arad\footremember{technion}{Physics Department, Technion, Israel}%
  \and David Gosset \footrecall{iqc} \footrecall{co}%
  }
\date{}

\begin{document}
\maketitle

\begin{abstract}
  Let $H$ be a frustration-free Hamiltonian describing a 2D grid of
  qudits with local interactions, a unique ground state, and local
  spectral gap lower bounded by a positive constant. For any
  bipartition defined by a vertical cut of length $L$ running from
  top to bottom of the grid, we prove that the corresponding
  entanglement entropy of the ground state of $H$ is upper bounded
  by $\tilde{O}(L^{5/3})$. For the special case of a 1D chain, our
  result provides a new area law which improves upon prior work, in
  terms of the scaling with qudit dimension and spectral gap.  In
  addition, for any bipartition of the grid into a rectangular
  region $A$ and its complement, we show that the entanglement
  entropy is upper bounded as $\tilde{O}(|\partial A|^{5/3})$ where
  $\partial A$ is the boundary of $A$. This represents the first
  subvolume bound on entanglement in frustration-free 2D systems.
  In contrast with previous work, our bounds depend on the local
  (rather than global) spectral gap of the Hamiltonian. We prove our results using a known method which bounds the
  entanglement entropy of the ground state in terms of certain
  properties of an approximate ground state projector (AGSP). To
  this end, we construct a new AGSP which is based on a robust
  polynomial approximation of the AND function and we show that it
  achieves an improved trade-off between approximation error and
  entanglement. 

\end{abstract}

\section{Introduction}

A regularly arranged collection of locally interacting spins may
hardly seem an accurate representation of the sea of molecules that
constitute a typical material. But the study of spin systems has
provided key insights into widely observed phenomena such as
ferromagnetism, superconductivity, superfluidity, and topological
order. Such insights have contributed to the technological progress
seen in materials science, electronics, and related areas.

Several universal features of quantum spin systems have been
discovered based on natural physical assumptions such as locality of
the interactions and/or the presence of a spectral gap in the
thermodynamic limit. Lieb and Robinson \cite{LiebR72} used locality
to conclude that, to a very good approximation, the support of local
observables expands at a constant rate as the system evolves in
time. For spin systems with a unique ground state and a spectral
gap, it has been shown that correlation functions decay
exponentially with distance \cite{Hastings04, HastingsK06,
NachtergaeleS06, aharonov2011detectability, GossetH15}. Hastings
\cite{Hastings07} proved that unique gapped ground states of
one-dimensional spin systems have bounded mutual information across
any bipartition of the lattice. This is the so-called \textit{area
law} for 1D quantum spin systems.

More generally, a quantum state of a system of qudits on a lattice
is said to obey an area law if for any bipartition of the lattice
into a region $A$ and its complement $\overline{A}$, the mutual
information between the parts scales as the size $|\partial A|$ of
the boundary of the bipartition. When the quantum state is pure, the
mutual information is twice the entropy of either bipartition, also
known as the \textit{entanglement entropy}. A quantum state
exhibiting an area law is markedly different from a random pure
quantum state, as the latter possesses an entanglement entropy that
scales as the volume of the smaller part. 

The study of the relationship between entanglement and geometry has
a long history \cite{eisert2008area}. Inspired by the work of
Bekenstein \cite{Bekenstein73} and Hawking
\cite{hawking1975particle}, which relates the entropy of a black
hole to its surface area, it was shown that the ground state of
several models of quantum field theories obey (or nearly obey) area
laws~\cite{bombelli86, srednicki1993, callan1994geometric,
KabatS94,HolzheyLW94}.  Later, a similar phenomenon was shown to
occur in the ground states of several systems with nearest-neighbor
interactions in one dimension \cite{audenaert2002entanglement,
latorre2003ground,VidalLRK03}, away from critical points where the
Hamiltonian becomes gapless and the entanglement may diverge. This
led to \textit{area law conjecture}, which states that the ground
states of gapped spin systems on a lattice of any dimension obey an
area law. 

This conjecture has led to a rich body of work connecting quantum
information science, condensed matter physics, and computer science.
Hastings' proof \cite{Hastings07} itself uses powerful
information-theoretic arguments inspired by the monogamy of quantum
entanglement \cite{Terhal01}. Brand{\~a}o and Horodecki
\cite{BrandaoH13} use ideas from the quantum communication task of
quantum state merging \cite{HorodeckiOW05} to obtain an area law for
any state satisfying an exponential decay of correlations in 1D. A
series of works \cite{AradLV12, AradKLV13} have obtained exponential
improvements to Hastings' entanglement upper bound, using the
\textit{polynomial method}, a widely used technique in theoretical
computer science and optimization theory.  These works have also led
to a rigorous proof that gapped ground states of 1D systems have an
efficient classical representation \cite{Hastings07, AradLV12,
AradKLV13} as Matrix Product States \cite{klumper1991equivalence,
klumper1992groundstate}, which explains the success of the DMRG
algorithm \cite{White93} in the numerical study of quantum spin
systems. The techniques developed in \cite{AradLV12, AradKLV13} have
also been used in the first provably efficient classical algorithm
for computing ground states of gapped 1D spin systems
\cite{LandauVV13, AradLVV17}. 

Despite these applications and extensions of Hastings' 1D result,
the area law conjecture for two (or higher) dimensional lattices has
thus far resisted all attacks.  Some works have described additional
physical assumptions which are sufficient to guarantee an area law.
For example, \Ref{WolfVHC08} proves an area law for the thermal
state of any spin system at sufficiently high temperatures.
Unfortunately, this bound diverges as the temperature approaches
zero, and hence does not provide any information about the ground
state. In \Ref{PhysRevA.80.052104} the area law was proved under the
assumption that the number of eigenstates with vanishing energy
density does not grow exponentially with the volume. In
\Ref{michalakis2012stability} it was proved under the assumption
that the Hamiltonian can be adiabatically connected to another
Hamiltonian in which the area-law holds, along a path of gapped
Hamiltonians. In \Ref{PhysRevLett.113.197204} the author assumed
that the ground state can be gradually built from a sequence of
ground states of smaller and smaller gapped Hamiltonians, and that
these ground states are, in some sense, very close to each other. In
\Ref{PhysRevB.92.115134} the area law was established under the
assumption of an exponential decay of the specific heat capacity of
the system with respect to the inverse temperature. A counterpart to the above positive results is provided by
\Ref{aharonov2014local}, which shows that a ``generalized" area law
for local Hamiltonian systems on arbitrary graphs is false.

In this work we establish a subvolume bound on the entanglement
entropy of the unique ground state of a frustration-free local
Hamiltonian in two dimensions with a local spectral gap. To state
our result, let us introduce some terminology. For the sake of being
concrete, we shall focus on a rather specific 2D setup. However,
many of the specific settings we assume can be easily generalized.

We consider a system of qudits of local dimension $d$ located at the
vertices of an $n\times L$ grid where $n\geq 2$ and $L\geq 1$, see
\Fig{fig:cut}.  We index qudits by their coordinates $(i,j)\in
[n]\times [L]$, where $[q]$ is the set of integers $\{1,2,\ldots q\}$. We define a
local Hamiltonian $H$ which acts on this system of qudits as a sum
of local projectors
\begin{align}
\label{def:2D-H}
  H = \sum_{i=1}^{n-1}\sum_{j=1}^{\max\{L-1,1\}} P_{ij} , 
\end{align}
where for $L\geq 2$, $P_{ij}$ is a projector ($P_{ij}^2=P_{ij}$)
that acts nontrivially only on the four qudits $\{(i,j),
(i,j+1),(i+1,j), (i+1,j+1)\}$. For the special case $L=1$,
\Eq{def:2D-H} describes a 1D chain $H=\sum_{i=1}^{n-1} P_{i1}$,
where $P_{i1}$ acts nontrivially only on qudits $i,i+1$. 

More generally, it will be convenient to view the system of qudits
as a 1D chain of ``columns". In particular, we define the
\textit{$i$th column} to be the set of qudits $\{(i,j): j\in [L]\}$,
and write the Hamiltonian as 
\begin{align}
\label{eq:twoloc}
  H=\sum_{i=1}^{n-1} H_i  \qquad \qquad  
    H_i\EqDef \sum_{j=1}^{\max\{L-1,1\}} P_{ij} ,
\end{align}
where the column Hamiltonian $H_i$ is the sum of all local
projectors which act nontrivially between qudits in columns $i$ and
$i+1$.

We further assume the Hamiltonian is frustration-free and has a
unique ground state $\ket{\Omega}$. Frustration-free means that the
ground state of the full Hamiltonian is also a ground state of each
of the individual local terms in the Hamiltonian, i.e., 
$P_{ij}\ket{\Omega}=0$ for all $i,j$.  This sort of Hamiltonian can
be viewed as a satisfiable instance of a \textit{quantum constraint
satisfaction} problem --- each local term is a constraint and the
ground state is a satisfying assignment \cite{bravyi2011efficient}.
Frustration-free quantum spin systems are widely studied in the
physics and quantum information literature (see e.g.,
Refs.~\cite{AffleckKLT87, HoussamLLNY19,PerezVWC08, SchuchCD10,
GottsteinW95,AlcarazSW95, KomaN97, Kitaev03, LevinW05}). 

The entanglement entropy of the ground state $|\Omega\rangle$ with
respect to some bipartition $[n]\times [L]=A\cup \overline{A}$ of the qudits is
\begin{align}
  S(\rho_A)\EqDef-\Tr(\rho_A \log(\rho_A)), 
    \qquad \text{ where} 
    \qquad \rho_A\EqDef\Tr_{\overline{A}}\ketbra{\Omega}.
\label{eq:entang}
\end{align}

Without any further assumptions, any nontrivial upper bound on
$S(\rho_A)$ must depend in some way on the spectral properties of
the Hamiltonian $H$. Indeed, ground states of gapless
frustration-free Hamiltonians in 1D can have very high entanglement
between the two halves of the chain~\cite{BravyiCMNS12,
MovassaghS16}, as large as the maximal linear scaling with chain
length~\cite{zhang2017entropy}.  The 1D area laws established in
Refs.~\cite{Hastings07, AradLV12, AradKLV13} depend on the (global) spectral gap of $H$, which for a
frustration-free Hamiltonian \Eq{eq:twoloc} is its smallest nonzero
eigenvalue. In contrast, the bounds we establish here depends on the
\textit{local spectral gap} $\gamma$ of $H$, equal to the minimum
spectral gap of any Hamiltonian describing a contiguous patch of the
system. In particular, for any contiguous subset $S\subseteq
[n]\times [L]$, define $\gamma(S)$ to be the smallest nonzero
eigenvalue of $\sum_{(i,j)\in S} P_{ij}$, and define
\begin{align*}
  \gamma=\min\{\min_{S} \text{$\gamma(S)$}, 1\}.
\end{align*}
Note that $0<\gamma \leq 1$.  It is slightly irksome that our results depend on the local rather
than the global spectral gap. The relationship between these two
quantities has been studied in Refs.~\cite{Knabe88, gosset2016local,
LemmM18, Lemm19}, see Section 2 of Ref. \cite{Lemm19} for a review. 
While in principle it is possible that the local gap is much smaller
than the global gap (potentially in exotic examples constructed in
\cite{CubittPW15, BauschCLP19}), we do not expect this to occur for
physically realistic systems. 

Our first result is a bound on the entanglement entropy of the
ground state with respect to a ``vertical cut" separating columns
$A=\{1,2,..., c\}$ from $\overline{A}=\{c+1,c+2,\ldots, n\}$ for
some $c\in[n-1]$. We will denote this vertical cut as $(c,c+1)$.
\begin{theorem}[\textbf{Subvolume scaling for a vertical cut}]
\label{thm:subv_cut}
  Let $\ket{\Omega}$ be the unique ground state of a
  frustration-free Hamiltonian \Eq{eq:twoloc} on an $n\times
  L$ grid of qudits with local dimension $d$. Its entanglement
  entropy across a vertical cut $(c,c+1)$ is at most
  \begin{align*}
    S(\rho_A)\leq \frac{CL^{5/3}}{\gamma^{5/6}} 
      \log^{7/3}(dL\gamma^{-1}).
  \end{align*}
where $C>0$ is a universal constant.
\end{theorem}

The above result can be viewed as simultaneously generalizing and
improving upon the previous state-of-the art area law in
1D~\cite{AradKLV13}.  Indeed, taking a grid of dimensions $n\times
1$ we recover the 1D case and Theorem \ref{thm:subv_cut} provides
the expected $O(1)$ bound on entanglement entropy for (locally)
gapped 1D systems, for which $d=O(1)$ and $\gamma=\Omega(1)$.
Looking more closely we see that \Thm{thm:subv_cut} improves upon
\Ref{AradKLV13} both in terms of the dependence on the local
dimension $d$, from $\log^3(d)$ to $\log^{7/3}(d)$, and in terms of
the dependence on the spectral gap $\gamma$, from $\gamma^{-1}$ to
$\gamma^{-{5/6}}$ (here ignoring a polylogarithmic factor as well as
the difference between local and global spectral gaps). 
This is a step closer
to the conjectured scaling of $\approx \frac{1}{\sqrt{\gamma}}$ for
1D frustration-free systems \cite{GossetH15, gosset2016local} which
coincides with the optimal upper bound on correlation length
\cite{GossetH15}.

The proof of \Thm{thm:subv_cut} is the main technical content of
this paper. The proof itself is essentially one dimensional in the
sense that it is entirely based on the expression in
\Eq{eq:twoloc} for the Hamiltonian as a 1D chain of columns.
With only a small modification we are able to establish a similar
bound for any bipartition of the 2D grid corresponding to a
rectangular region and its complement. The bound is obtained by
viewing the Hamiltonian as a 1D chain of concentric rectangular
bands and using almost exactly the same proof, see
\Fig{fig:transformation} (c). 

\begin{theorem}[\textbf{Subvolume scaling for a rectangular region}]
  Let $\ket{\Omega}$ be the unique ground state of a
  frustration-free Hamiltonian \Eq{eq:twoloc} on an $n\times L$ grid
  of qudits with local dimension $d$. Its entanglement entropy with
  respect to a bipartition of the qudits into a rectangular region
  $A$ and its complement $\overline{A}$ is given by
  \begin{align}
    S(\rho_A) \leq \frac{C|\partial A|^{5/3}}{\gamma^{5/6}} 
      \log^{7/3}(d|\partial A|\gamma^{-1}).
\label{eq:sboundrec}
  \end{align}
  where $C>0$ is a universal constant.
\label{thm:subr_cut}
\end{theorem}
The bound $\tilde{O}(|\partial A|^{5/3})$ on the right-hand side
represents an improvement over the trivial volume law scaling
of $|\partial A|^2$, and gives some movement towards the elusive
area law conjecture in two dimensions. 

While Theorems~\ref{thm:subr_cut} and~\ref{thm:subv_cut} are stated
in terms of the usual entanglement entropy of the ground state, we
are able to obtain similar subvolume bounds on other so-called
R{\'e}nyi entanglement entropies. In Appendix \ref{app:huangbound}
we show how this works for the R{\'e}nyi entanglement entropy of
order $1/2$.  An interesting consequence is then obtained following
an argument from the recent works in Refs.~\cite{Huang19,
DalzellB19}. Theorem~3 of \Ref{Huang19} shows that, if some
R{\'e}nyi entanglement entropy (of order less than one) of a quantum
state $\ket{\psi}$ on a 2D lattice satisfies an area law for any
bipartition of the lattice into a square region and its complement,
then there is a projected entangled pair (PEPS) state of bond
dimension $e^{O(\frac{1}{\delta})}$ which reproduces expectation
values of all local observables in the state $\psi$ up to an
additive error $\delta$. Following Huang's proof technique and using
our subvolume law one can reach almost the same conclusion --- but
with a weaker upper bound $e^{\tilde{O}(\frac{1}{\delta^5})}$ on the
bond dimension, see Appendix~\ref{app:huangbound} for details. 

Finally, we remark that it may be possible to extend our results to
degenerate ground states, using the techniques developed in
\Ref{AradLVV17}, although we do not pursue this direction here. 

To prove \Thm{thm:subv_cut} we use a method described in
\Ref{AradKLV13} which is based on the construction of a so-called
Approximate Ground State Projector (AGSP). In this context an AGSP
is an operator $K$ which fixes the ground state and its orthogonal
complement, i.e., 
\begin{align*}
  K\ket{\Omega} = K^{\dagger}\ket{\Omega}=\ket{\Omega} .
\end{align*}
The AGSP has two important parameters $D$ and $\Delta$ which are
defined with respect to a given bipartition of the qudits. The
parameter $D$ is an upper bound on the Schmidt rank of $K$ across
the bipartition. Recall that the Schmidt rank of an operator $K$
acting on two registers $A$ and $B$ is the smallest integer $R$ such
that $K= \sum_{s=1}^R K^s_A\otimes K^s_B$ for some operators
$\{K^s_A\}_{s=1}^R$ and $\{K^s_B\}_{s=1}^R$ that are supported only
on the registers $A$ and $B$, respectively. The parameter $\Delta$
is any number such that
\begin{align*}
  \norm{K\ket{\psi}}^2\leq \Delta 
    \quad \text{ for all } \quad \ket{\psi}\in G_{\perp},
\end{align*}
where $G_\perp$ is the subspace of $nL$-qudit states orthogonal to
the ground state $\ket{\Omega}$. In other words, $\Delta$ is a
\textit{shrinking factor} which measures the shrinkage of the space
orthogonal to $|\Omega\rangle$ when $K$ is applied.  An AGSP with
parameters $D$ and $\Delta$ is called a $(D,\Delta)$-AGSP. The
following theorem relates these AGSP parameters to a bound on the
entanglement entropy across the cut.
\begin{theorem}[\Ref{AradLV12}]
\label{thm:AGSParealaw}
  If there exists a $(D,\Delta)$-AGSP such that
  $D\cdot \Delta\le \frac{1}{2}$, then the entanglement entropy of
  $\ket{\Omega}$ across the cut is upper bounded by
  $10\cdot\log(D)$.
\end{theorem}

\Thm{thm:AGSParealaw} states that the existence of an AGSP with the
right parameters implies a bound on the entanglement entropy of the
ground state $\ket{\Omega}$ of our quantum spin system. Most of our
work in the remainder of the paper will be to establish bounds on
the parameters $D,\Delta$ of a certain AGSP. At a high level, the AGSP we construct in this paper is based on the
detectability lemma operator introduced in \Ref{AharonovALV08} and
its coarse-grained version used in Refs.~\cite{AradLV12, AAV16}. We
are able to improve upon its performance in terms of the parameters
$D$ and $\Delta$ by modifying the construction using certain
polynomial approximations.

The remainder of the paper is organized as follows. In
\Sec{sec:preliminaries} we describe two families of polynomials.
These are building blocks used to construct the AGSP, which is our
main object of study, given in \Sec{sec:agsp}. We also include a
sketch of the proof of Theorem \ref{thm:subv_cut} in \Sec{sec:agsp}.
In Sections~\ref{sec:shrinking} and \ref{sec:schmidt} respectively
we upper bound the shrinking factor $\Delta$ and Schmidt rank $D$ of
this AGSP. In \Sec{sec:arealaw1D} we combine \Thm{thm:AGSParealaw}
and the bounds on $D$ and $\Delta$ to complete the proof of 
\Thm{thm:subv_cut}. Finally, in \Sec{sec:arealawrectangle} we
describe the minor modifications to the proof which result in
\Thm{thm:subr_cut}.

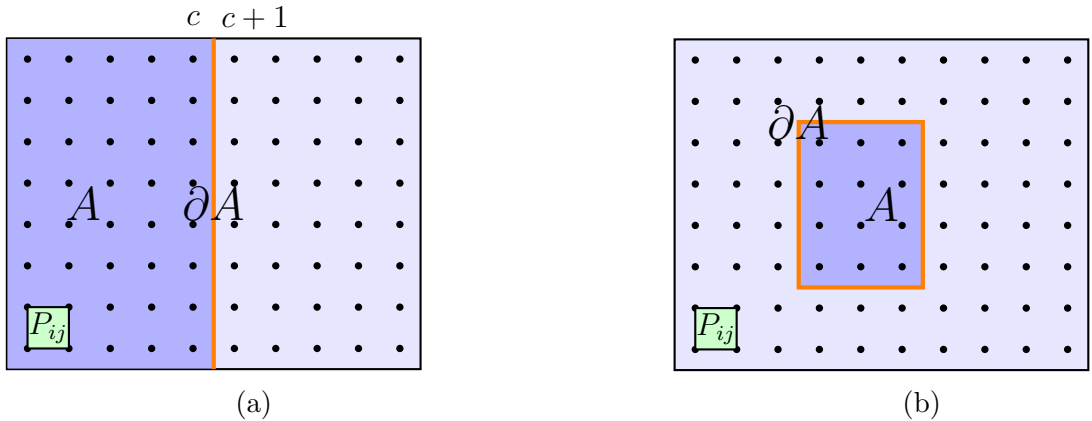
\begin{figure}
\centering
\begin{subfigure}[b]{0.4\textwidth}
\begin{tikzpicture}[xscale=0.55,yscale=0.55]

\draw [fill=blue!10!white, thick] (0.5,0.5) rectangle (10.5, 8.5);
\draw [fill=blue!30!white] (0.5, 0.5) rectangle (5.5,8.5);
\draw [orange, ultra thick] (5.5, 8.5)--(5.5,0.5);
\node at (5, 9){$c$};
\node at (6.5, 9){$c+1$};
\node[scale=1.5] at (2.35,4.5) {$A$};
\foreach \i in {1,...,10}
{
\foreach \j in {1,...,8}
   \draw (\i, \j) node[circle, fill=black, scale=0.25]{};
}

\draw [fill=white!80!green, thick] (1,1) rectangle (2,2);
\node at (1.5, 1.5) {$P_{ij}$};
\node[scale=1.5] at (5.5,4.5) {$\partial A$ };
\end{tikzpicture}
\caption{}
  \end{subfigure}
\hspace{2cm}
\begin{subfigure}[b]{0.4\textwidth}
\begin{tikzpicture}[xscale=0.55,yscale=0.55]

\draw [fill=blue!10!white, thick] (0.5,0.5) rectangle (10.5, 8.5);
\draw [fill=blue!30!white] (3.5, 2.5) rectangle (6.5,6.5);
\draw [orange, ultra thick] (3.5, 2.5) rectangle (6.5,6.5);
\node[scale=1.5] at (3.5,6.5) {$\partial A$};
\foreach \i in {1,...,10}
{
\foreach \j in {1,...,8}
   \draw (\i, \j) node[circle, fill=black, scale=0.25]{};
}

\draw [fill=white!80!green, thick] (1,1) rectangle (2,2);
\node at (1.5, 1.5) {$P_{ij}$};
\node[scale=1.5] at (5.5,4.5) {$A$ };
\end{tikzpicture}
  \caption{}
\end{subfigure}
\caption{Bipartitions considered in (a) Theorem \ref{thm:subv_cut}
and (b) Theorem \ref{thm:subr_cut} \label{fig:cut}}
\end{figure}

\section{Polynomials}
\label{sec:preliminaries}

Here we describe two families of polynomials, which are the building
blocks for our AGSP.

We first describe a univariate polynomial function of $x$ that takes
the value $1$ at $x=0$ but has a very small magnitude in some range
of $x$-values bounded away from $0$.  We shall colloquially refer to
this as a \textit{step polynomial}. It is well-known that Chebyshev
polynomials can be used for this purpose. Let $T_f$ be the
degree-$f$ Chebyshev polynomial of the first kind. For any positive
integer $f$ and $g\in (0,1)$, define
\begin{align}
  \mathrm{Step}_{f,g}(x) \EqDef 
    \frac{T_f\left(\frac{2(1-x)}{1-g}-1\right)}
      {T_f\left(\frac{2}{1-g}-1\right)}.
\label{eq:shiftedchebyshev}
\end{align}
The following fact is a special case of Lemma~4.1 from
\Ref{AradKLV13}\footnote{Fact~\ref{fact:polyD} is obtained from
Lemma 4.1 of \Ref{AradKLV13} by setting $\epsilon_0=0, \epsilon_1=g,
u=1, \ell=f$.}.
\begin{fact}[\textbf{Step polynomial} \cite{AradKLV13}]
\label{fact:polyD} 
  For every positive integer $f$ and $g\in (0,1)$, there exists a
  univariate polynomial $\mathrm{Step}_{f,g}:
  \mathbbm{R}\rightarrow\mathbbm{R}$ with real coefficients and degree
  $f$ such that $\mathrm{Step}_{f,g}(0)=1$ and
  \begin{align}
  \label{eq:step-poly}
   \left|\mathrm{Step}_{f,g}(x)\right| \leq 2\exp(-2f\sqrt{g})
      \qquad \text{for} \qquad g\leq x\leq 1.
  \end{align}
\end{fact}
The only properties of $\mathrm{Step}_{f,g}$ that we will use in the
following are summarized in Fact \ref{fact:polyD}. In particular, we
will not need its precise form~\eqref{eq:shiftedchebyshev}, which is
included only for completeness.

A key ingredient in our work is the construction of \textit{robust
polynomials} due to Sherstov~\cite{Sherstov12}. In this setting a
robust polynomial is a real multivariate polynomial that approximates a boolean function even when the input
$x\in \{0,1\}^m$ is corrupted by a real-valued error vector
$\epsilon\in [-1/20,1/20]^m$. We will use a robust polynomial for
the `AND' function on $m$ variables which has properties summarized
in the following theorem. The construction of this polynomial and
the proof of the theorem, which is provided in
Appendix~\ref{append:ANDpoly}, follow the technique used by Sherstov
in Theorem 3.2 of Ref.~\cite{Sherstov12} to construct a robust
polynomial approximation of the PARITY function. 
\begin{theorem}[\textbf{Robust AND polynomial}, following 
  Sherstov~\cite{Sherstov12}]
\label{corrobust}
  Let $m$ be a positive integer. There is a multivariate polynomial
  $p_{\mathrm{AND}}:\mathbbm{R}^m\rightarrow \mathbbm{R}$ with real
  coefficients and degree $11m$ satisfying \begin{align*}
  p_{\mathrm{AND}}(1,1,...,1)=1, \end{align*} such that for any bit-string
  $y\in \{0,1\}^m$ and real-valued error vector $\epsilon\in
  [-1/20,1/20]^m$, we have
  \begin{align}
    \label{eq:robusteq}
    \left|p_{\mathrm{AND}}(y+\epsilon)-y_1y_2\ldots y_m\right|
      \leq e^{-m}.
  \end{align}
  Moreover, there are univariate polynomials
  $A_i:\mathbbm{R}\rightarrow \mathbbm{R}$ of degree $2i+1$ for each
  $i\geq 0$ such that
  \begin{align*}
    p_{\mathrm{AND}}(x_1, \ldots x_m) 
      = \sum_{\{i_1, \ldots i_m\}: i_1+\ldots +i_m \le 5m}
        A_{i_1}(x_1)\cdot A_{i_2}(x_2)\ldots A_{i_m}(x_m).
  \end{align*}
\end{theorem}

\section{Approximate Ground State Projector}
\label{sec:agsp}

Let us begin by defining a simple AGSP as our starting point. The
AGSP depends on a positive integer $t$, which is a coarse-graining
parameter. For any $2t\leq k\leq n-2t$ define $Q'_k$ as the
projection onto the ground space of the Hamiltonian 
\begin{align*}
  h'_k \EqDef \sum_{j=k-2t+1}^{k+2t-1} H_j,
\end{align*}
which contains all terms of \Eq{eq:twoloc} supported entirely inside
the contiguous region of the $4t$ columns $\{k-2t+1,k-2t+2,\ldots,
k+2t\}$. Here we use a prime superscript because we will
soon slightly modify the notation for the subregion Hamiltonians
$h'_k$ and ground space projectors $Q'_k$. Note that
\begin{align}
  [Q'_a,Q'_b]=0 \quad \text{whenever} \quad |a-b|\geq 4t,
\label{eq:disjoint}
\end{align}
as the latter condition ensures the projectors have disjoint support.

Let us define the ($t$-coarse grained) \textit{Detectability lemma
operator} \cite{AharonovALV08} to be the product
\begin{align}
  \DL(t) \EqDef \br{Q'_{2t}\cdot Q'_{8t}\cdot Q'_{14t}
    \cdot\ldots}\cdot \br{Q'_{5t}\cdot Q'_{11t}
      \cdot Q'_{17t}\cdot\ldots} \label{eq:DL1} ,
\end{align}
where the terms within each of the parenthesized expressions are
mutually commuting. The following Lemma is a slight variant of one
established in \Ref{AAV16}. We provide a proof in \App{append:DLproof}.  Recall that $G_\perp$
is the subspace orthogonal to the ground state.
\begin{lemma}
  \label{lem:DL}
  For any normalized state $\ket{\psi}\in G_\perp$,
  \begin{align*}
    \norm{\DL(t)\ket{\psi}}
      \le 2e^{-t\sqrt{\gamma}/25} .
  \end{align*}
\end{lemma}
The lemma states that, if the coarse-graining parameter $t$ is large
enough, then the operator $\DL(t)$ shrinks the space $G_{\perp}$ at
a rate which decreases exponentially with the \textit{square root}
of $\gamma$. This square-root is the reason why coarse-graining is
useful to us---without it the shrinkage would be quadratically worse
as a function of $\gamma$ (see, e.g., \Ref{AAV16}).

Recall that we are interested in the entanglement of the ground
state $\ket{\Omega}$ across some vertical cut $(c,c+1)$ where $c\in
[n]$. For now it will be convenient to assume that $c \mod 6t=2t$;
later, in \Sec{sec:arealaw1D}, we drop this assumption. In
this case the set of qudits $\{c,c+1\}\times [L]$ are contained in
the support of $Q'_c$ and the cut divides its support into two equal
parts, see Fig. \ref{fig:agsp}. Moreover, $Q'_c$ is the only projector in \Eq{eq:DL1} with
support intersecting this vertical cut.  It will be convenient to
rewrite \Eq{eq:DL1} using a different notation which singles out
some of the projectors that surround the cut. In particular, let $m$
be an odd positive integer and consider the $m$ projectors 
\begin{align}
  Q'_{c-3(m-1)t}, \ldots, Q'_{c-6t}, Q'_c, Q'_{c+6t}, 
    \ldots, Q'_{c+3(m-1)t}.
\label{eq:projold}
\end{align}
Recall that these operators project onto the ground spaces of
subregion Hamiltonians
\begin{align}
  h'_{c-3(m-1)t}, \ldots, h'_{c-6t}, h'_c, h'_{c+6t}, 
    \ldots, h'_{c+3(m-1)t}.
\label{eq:hamold}
\end{align}
We shall relabel the projectors \Eq{eq:projold} from left-to-right
as $Q_1,Q_2, \ldots, Q_m$ and the corresponding subregion
Hamiltonians \Eq{eq:hamold} as $h_1,h_2, \ldots, h_m$. Then $Q_k$ is
the ground space projector of $h_k$ for each $k=1,2,\ldots, m$.
Observe that
\begin{align*}
[Q_i, Q_j]=0 \quad \text{for all} \quad 1\leq i\leq j\leq m,
\end{align*}
which follows from \Eq{eq:disjoint} and our definition of the
projectors $Q_1,Q_2\ldots, Q_m$. We write
\begin{align}
  \DL(t)=Q_1Q_2\ldots Q_m Q_{\mathrm{rest}} ,
\label{eq:DL2}
\end{align}
where $Q_{\mathrm{rest}}$ contains the remainder of the terms in
\Eq{eq:DL1}, i.e., 
\begin{align*}
  Q_{rest}\EqDef \br{Q'_{2t}\cdot \ldots\cdot
    Q'_{c-3(m-1)t-6t}\cdot Q'_{c+3(m+1)t+6t}\cdot 
    \ldots}\br{Q'_{5t}\cdot Q'_{11t}\cdot Q'_{17t}\cdot\ldots}.
\end{align*}
Note that the difference between Eqs.~(\ref{eq:DL1},~\ref{eq:DL2})
is only notation and that the operator $\DL(t)$ does not have any
dependence on the parameter $m$. However, we will soon use
\Eq{eq:DL2} as a starting point in defining another AGSP which does
depend on this parameter.

\begin{figure}
\begin{tikzpicture}[xscale=0.4,yscale=0.6]
\foreach \k in {0,...,5}
{
\draw[draw=black, fill=green!30!white] (0.5+6*\k,-2) rectangle (4.5+6*\k,-1);

\draw[draw=black, fill=red!30!white] (3.5+6*\k,2) rectangle (7.5+6*\k,1);
}

\draw[draw=black, fill=green!30!white] (0.5+6*6,-2) rectangle (4.5+6*6,-1);
\draw[thick, |<->|] (3.5+6*2, 2.5)--(7.5+6*2, 2.5);
\draw (17.5, 3) node {$4t$};

\draw[thick, |<->|] (3.5+6*1,2.5)--(4.5+6*1,2.5);
\draw (10,3) node {$t$};
\draw (10,4) node {Overlap region};

\draw (20, 0) node[circle, fill=red, scale=0.5]{};

\draw (25.5,-1.5) node {$Q_3$};
\draw (19.5,-1.5) node {$Q_2$};
\draw (13.5,-1.5) node {$Q_1$};

\draw (19.9,-0.6) node {$c$};

\draw (22,-0.6) node {$c+1$};

\foreach \i in {1,...,40}
{
   \draw (\i, 0) node[circle, fill=black, scale=0.5]{};
}

\
\foreach \i in {14,...,15}
{
\draw (\i, 0) node[circle, fill=blue, scale=0.5]{};
}
\foreach \i in {20,...,21}
{
\draw (\i, 0) node[circle, fill=blue, scale=0.5]{};
}
\foreach \i in {26,...,27}
{
\draw (\i, 0) node[circle, fill=blue, scale=0.5]{};
}

\draw[thick, dotted] (20.5,3)--(20.5,-3);

\end{tikzpicture}
  \caption{Construction of the AGSP $K(m,t,\ell)$, for $t=1$ and
  $m=3$. Here each dot represents a column of $L$ qudits; the 2D
  $n\times L$ grid is represented as a 1D chain of $n$ columns. Each
  coarse-grained projector contains $4t$ columns in its support and
  neighboring coarse-grained projectors overlap in a region
  containing $t$ columns. The $m$ coarse-grained projectors around
  the cut are denoted $Q_1,Q_2,\ldots, Q_m$. The columns shown in
  blue are members of the set `$\Ind$' defined in Section
  \ref{sec:schmidt}.\label{fig:agsp}}
\end{figure}
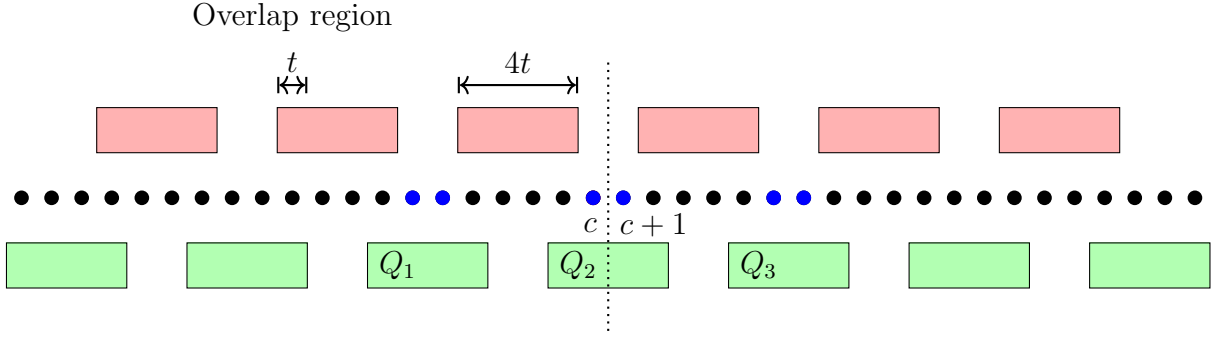

To this end, we shall first define polynomial approximations to each
of the projectors $Q_1,Q_2,\ldots, Q_m$. In particular, we use the
degree-$f$ \textit{step polynomial} $\mathrm{Step}_{f,g}(x)$ of
Fact~\ref{fact:polyD} to define
\begin{align}
\label{def:Qhat}
  \widehat{Q}_j 
    \EqDef \mathrm{Step}_{f,g}\left(\frac{1}{4tL} h_j\right)
    \qquad \text{where} \quad f
      = \left\lceil4\sqrt{tL/\gamma}\right\rceil
    \quad \text{and} \quad g=\frac{\gamma}{4tL}.
\end{align}
Here $\lceil x \rceil$ indicates the smallest integer which is at
least $x$. Note that $\widehat{Q}_j$ is Hermitian, since
$\mathrm{Step}_{f,g}$ is a polynomial with real coefficients, and
$h_j$ is a Hermitian operator. In addition, $\|h_j/4tL\|\leq 1$
since $h_j$ is a sum of at most $4tL$ projectors, and the smallest
nonzero eigenvalue of $h_j/4tL$ is at least $\gamma/4tL$ (by
definition of the local spectral gap $\gamma$). We use Fact~\ref{fact:polyD}
to establish the following properties of the spectrum of
$\widehat{Q}_j$.
\begin{lemma}
  \label{lem:qeig} 
  For each $j=1,2,\ldots, m$, the projector onto the eigenspace of
  $\widehat{Q}_j$ with eigenvalue $+1$ is equal to $Q_j$, and
  \begin{align}
    \|\widehat{Q}_j-Q_j\|
      \leq \frac{1}{20}.
    \label{eq:qeig}
  \end{align}
\end{lemma}
\begin{proof}
Recall that $Q_j$ projects onto the zero energy ground space of
$h_j$, which is mapped to the $+1$-eigenspace of $\widehat{Q}_j$
since $\mathrm{Step}_{f,g}(0)=1$. On the other hand, all nonzero
eigenvalues of $h_j$ are at least $\gamma/4tL$ and using
\Eq{eq:step-poly} and the choices \Eq{def:Qhat} of $f$ and $g$ we
see that 
\begin{align*}
  \|\widehat{Q}_j-Q_j\|\leq 2e^{-2f\sqrt{\gamma/4tL}}
    \leq 2e^{-4}\leq \frac{1}{20}.
\end{align*}
\end{proof}
\noindent The AGSP we will define is similar to the $\DL(t)$
operator of \Eq{eq:DL2}, with the following modifications:
\begin{itemize}

  \item \textbf{Outer polynomial approximation using the robust AND:}
   Use $p_{\mathrm{AND}}$ described in \Thm{corrobust} to
   approximate the product $Q_1Q_2\ldots Q_m$.

  \item \textbf{Inner polynomial approximation using the step 
    polynomials:} 
    Use the operators $\widehat{Q}_j$ to approximate $Q_j$. 
    
  \item \textbf{Powering:} 
    Amplify the effect of the operator by raising it to a power
    $\ell\geq 1$.
\end{itemize}

In particular, the AGSP which is the main object of study in this
paper is defined as follows:
\begin{align}
\label{def:Kmtl}
\boxed{  K(m,t,\ell) \EqDef \left(p_{\mathrm{AND}}
    \left(\widehat{Q}_1,\widehat{Q}_2,\ldots, 
    \widehat{Q}_m\right) Q_{\mathrm{rest}}\right)^\ell.}
\end{align}
It depends on the choices of coarse-graining parameter $t$ (a
positive integer), the odd positive integer $m$ describing the
number of coarse-grained projectors of interest near the cut, and the positive integer $\ell$ which is
the powering parameter.  Note that if $n$ is too small we may not be
able to fit $m$ coarse-grained projectors around the cut as shown in
\Fig{fig:agsp}, and in this case strictly speaking we cannot define
$K(m,t,\ell)$ as above. In the following, we shall without loss of
generality assume that $n$ is sufficiently large that $K(m,t,\ell)$
is well-defined \footnote{We can always form a new Hamiltonian $H'$
on an $n'\times L$ grid for any $n'>n$ which has (a) the same local
spectral gap $\gamma$ as $H$, and (b) a unique ground state
$|\Omega\rangle\otimes |0^{(n'-n)L}\rangle$ and therefore exactly
the same entanglement entropy across the given cut. $H'$ is obtained
from $H$ by adding new local projectors which act on all the newly
added plaquettes of the lattice. For each plaquette with $q<4$ old
qudits from the original lattice and $4-q$ new qudits, we add the
projector $I^{\otimes q}\otimes (I-|0\rangle\langle0|^{\otimes
4-q})$ to the Hamiltonian.}.  

To confirm that the operator $K(m,t,\ell)$ is an AGSP, we need to
check that it fixes the ground state $|\Omega\rangle$ and its
orthogonal space $G_\perp$, that is, 
\begin{align}
     K(m,t,\ell)^{\dagger}\ket{\Omega}
       = K(m,t,\ell)\ket{\Omega} = \ket{\Omega}
    \label{eq:Kpart2}
  \end{align}
It suffices to check \Eq{eq:Kpart2} with $\ell=1$ since the result
for higher $\ell$ follows from this special case. Using the fact
that $Q_{\mathrm{rest}}\ket{\Omega}=\ket{\Omega}$ we get
\begin{align*}
  K(m,t,1)\ket{\Omega}=p_{\mathrm{AND}}\left(\widehat{Q}_1,
    \widehat{Q}_2,\ldots, \widehat{Q}_m\right)\ket{\Omega}
    = p_{\mathrm{AND}}(1,1,\ldots, 1)\ket{\Omega}=\ket{\Omega},
\end{align*}
where in the second equality we used the fact that
$\widehat{Q}_j\ket{\Omega}=\ket{\Omega}$ for all $j=1,2,\ldots, m$,
and in the last equality we used the fact that
$p_{\mathrm{AND}}(1,1,\ldots, 1)=1$ as stated in Thm{corrobust}. A
very similar argument shows $ K(m,t,1)^{\dagger}
\ket{\Omega}=\ket{\Omega}$.

In the next two sections we bound the shrinking factor $\Delta$ and
Schmidt rank $D$ of the AGSP $K(m,t,\ell)$ across the vertical cut
$(c,c+1)$.  We now provide an overview of these bounds and how they
are used to establish Theorem \ref{thm:subv_cut}.

\subsection*{Overview of the proof of Theorem \ref{thm:subv_cut}:}

In \Sec{sec:shrinking} we use the error bound
\Eq{eq:robusteq} for the robust polynomial $p_{\mathrm{AND}}$
to show that $K(m,t,1)$ approximates the coarse-grained
detectaibility lemma operator $\DL(t)$ in the sense that
$\norm{K(m,t,1)-\DL(t)}\leq e^{-m}$. In particular, choosing 
\begin{align}
  t=\Theta(m\gamma^{-1/2}),
\label{eq:tas}
\end{align}
is enough to ensure that the shrinking factor $\Delta$ of
$K(m,t,\ell)$ is asymptotically the same as that of
$(\DL(t))^{\ell}$ (from Lemma \ref{lem:DL}) , that is
\begin{align}
  \Delta=e^{-\Omega(m\ell)},
\label{eq:shrink1}
\end{align}
see Theorem \ref{thm:shrink}. 

Next, we need to understand the Schmidt rank $D$ of $K(m,t,\ell)$.
Fixing $t$ as in \Eq{eq:tas}, in \Sec{sec:schmidt} we
show that if
\begin{align}
  \ell =\Theta(m^{5/2}L^{1/2}\gamma^{-1/4}),
\label{eq:ellas}
\end{align}
then we have the upper bound
\begin{align}
  D = e^{\tilde{O}\left(mtL+\ell\right)}
    = e^{\tilde{O}(m^2L\gamma^{-1/2}+\ell)},
\label{eq:rank1}
\end{align}
see \Thm{thm:sr} (below we give some high level explanation for
\Eq{eq:rank1}). We then choose $m$ to satisfy $D\cdot
\Delta<1/2$ so that \Thm{thm:AGSParealaw} can be applied.
Comparing Eqs.~(\ref{eq:shrink1}, \ref{eq:ellas},\ref{eq:rank1}) we
see that this leads to
\begin{align}
  m=\tilde{\Theta}(L^{1/3}\gamma^{-1/6}).
\label{eq:mtilde}
\end{align}
The entanglement entropy of the ground state $\ket{\Omega}$ is
upper bounded using \Thm{thm:AGSParealaw} as 
\begin{align*}
  10\log(D)=\tilde{O}(L^{5/3}\gamma^{-5/6}) 
\end{align*} 
as claimed in \Thm{thm:subv_cut}. Here we are hiding factors
polylogarithmic in $d,L, \gamma^{-1}$ in the $\tilde{O}(\cdot)$
notation, while in \Sec{sec:arealaw1D} we give a more explicit proof
which carries them around.

The most involved technical component of this work is to establish
the bound \Eq{eq:rank1} on $D$. We use a variant of an argument from
\Ref{AradKLV13}, which can be understood at a high level as follows.
Imagine starting with the definition~\eqref{def:Kmtl} of our AGSP and
then expanding the degree-$11m$ polynomial $p_{\mathrm{AND}}$ and
the degree-$f$ polynomials $\{\widehat{Q}_j\}$ where $f$ is given by
\Eq{def:Qhat}. Looking at the total degree of the polynomials we are
expanding and multiplying by the power $\ell$, we see that
$K(m,t,\ell)$ can be written as a sum of terms, each of which is a
product $P$ of at most $O(m\cdot f\cdot \ell)$ operators from the
set $Q_{rest}\cup \{H_i\}_{i\in \Loc}$, where $\Loc\subset [n]$ is a
set of $\sim 4mt$ column indices centered around the cut of
interest. Consider now a \emph{single} such product $P$ in the expansion. 
Since $|\Loc|=O(mt)$, we can always find an index $k\in \Loc$ such
that the number of times $H_k$ occurs in $P$ is at most
$O(mf\ell/(mt))=O(f\ell/t)$. Therefore $P$ has Schmidt rank at most
$e^{\tilde{O}(\frac{f\ell}{t})}$ across the cut $(k,k+1)$. Since $k$
is at most $4mt$ columns away from $c$, and since each column
contains $L$ qudits, the operator $P$ has Schmidt rank at most 
\begin{align}
  e^{\tilde{O}(\frac{f\cdot \ell}{t}+mtL)}
\label{eq:pbound}
\end{align}
across the cut $(c,c+1)$ of interest. With our choice of $m,t,\ell$
given in Eqs.~(\ref{eq:mtilde},\ref{eq:tas},\ref{eq:ellas}) and with
$f$ given by \Eq{def:Qhat}, one can confirm that the expression
\Eq{eq:pbound} coincides with \Eq{eq:rank1} which is the bound we
are trying to establish. Unfortunately, \Eq{eq:pbound} is only an
upper bound on the Schmidt rank of each product $P$, while we are
interested in an upper bound on the Schmidt rank of $K(m,t,\ell)$
which is a sum of many such products. It turns out that naively
bounding the latter quantity by the number of products times
\Eq{eq:pbound} is not good enough to obtain the desired result. In
other words, the only problem with the above proof technique is that
the decomposition of the AGSP as a sum of products $P$ has too many
such terms. In \Sec{sec:schmidt} we prove the bound \Eq{eq:rank1}
using a variant of the above strategy which is based on an expansion
of $K(m,t,\ell)$ as a sum of (far fewer) well-structured operators
of a certain form, which take the place of the products $P$
considered above. Since we were initially guided by the back-of-the
envelope estimate \Eq{eq:pbound}, it is fortunate that the actual
proof is close enough in spirit that it provides the same asymptotic
bound on Schmidt rank of our AGSP.

\section{ Shrinking factor of the AGSP}
\label{sec:shrinking}
In this Section we use the properties of the robust polynomial
$p_{\mathrm{AND}}$ summarized in \Thm{corrobust} to upper bound the
shrinking factor $\Delta$ of our AGSP.
\begin{theorem}[\textbf{AGSP shrinking bound}]
\label{thm:shrink}
  Let $\ket{\psi}\in G_\perp$ be a normalized state.
  Then for all $\ell\geq 1$ we have
  \begin{align}
    \norm{K(m,t,\ell)\ket{\psi}}^2\le \Delta 
    \qquad \text{where} \qquad 
    \Delta=\left(e^{-m}
      + 2e^{-t \sqrt{\gamma}/25}\right)^{2\ell}.
  \label{eq:Kpart1}
\end{align}
   \end{theorem}

\begin{proof}
  Note that it suffices to prove the claim for $\ell=1$ since the
  result for higher $\ell$ follows straightforwardly from this
  special case. Recall that the Hermitian operators
  $\widehat{Q}_1,\widehat{Q}_2,\ldots, \widehat{Q}_m$ mutually
  commute and therefore can be simultaneously diagonalized. Since
  $p_{\mathrm{AND}}$ is a polynomial with real coefficients, 
  $p_{\mathrm{AND}}\left(\widehat{Q}_1,\widehat{Q}_2,\ldots,
  \widehat{Q}_m\right)$ is a Hermitian operator. Let us write
  $\Pi_j^{(x)}$ for the projector onto the eigenspace of
  $\widehat{Q}_j$ with eigenvalue $x$. Note that Lemma
  \ref{lem:qeig} states that $\Pi_{j}^{(1)}=Q_j$ and all eigenvalues
  of each operator $\widehat{Q}_j$ lie in the range
  \begin{align}
  \label{eq:epsrange}
    x \in [-1/20,1/20]\cup \{1\}.
  \end{align}
  Thus
  \begin{align}
    &p_{\mathrm{AND}}(\hQ_1, \hQ_2, \ldots \hQ_{m}) 
    = \sum_{x_1, x_2, \ldots x_{m}}
      p_{\mathrm{AND}}\br{x_1, x_2, \ldots x_m}
        \Pi_1^{(x_1)}\Pi_2^{(x_2)}\ldots 
        \Pi_{m}^{(x_{m})}\nonumber\\
    &= p_{\mathrm{AND}}(1,1,\ldots 1)Q_1Q_2\ldots Q_m 
      +\sum_{\stackrel{x_1, x_2, \ldots 
        ,x_m:}{ \exists i \text{ with } x_i\in [-1/20, 1/20]}}
          p_{\mathrm{AND}}\br{x_1, x_2, \ldots x_m}
          \Pi_1^{(x_1)}\Pi_2^{(x_2)}
            \ldots \Pi_m^{(x_m)}
  \label{eq:sump}
  \end{align}
  Using \Thm{corrobust} and \Eq{eq:epsrange} we bound each term
  appearing the sum on the right-hand-side as
  \begin{align}
  \label{eq:killx}
    \left|p_{\mathrm{AND}}\br{x_1, x_2, \ldots x_m} \right|
      \leq e^{-m} 
        \qquad \text{ whenever} \qquad  \exists i \text{ with }
          x_i \in [-1/20, 1/20].
  \end{align}
  Therefore, using the fact that $p_{\mathrm{AND}}(1,1,1\ldots,
  1)=1$ in \Eq{eq:sump} and the mutual orthogonality of the
  operators $\{ \Pi_1^{(x_1)}\Pi_2^{(x_2)}\ldots
  \Pi_m^{(x_m)}\}_{x_1,x_2,\ldots x_m}$, we get
  \begin{align*}
    \norm{p_{\mathrm{AND}}(\hQ_1, \hQ_2, \ldots \hQ_{m}) -Q_1Q_2\ldots Q_m}
      \le e^{-m},
  \end{align*}
  and so
  \begin{align}
    \norm{K(m,t,1)-\DL(t)} 
      &\le \Big\|\left(p_{\mathrm{AND}}(\hQ_1, \hQ_2, \ldots \hQ_{m}) 
        - Q_1Q_2\ldots Q_m\right)Q_{\mathrm{rest}}\Big\|\nonumber\\
    &\le \norm{p_{\mathrm{AND}}(\hQ_1, \hQ_2, \ldots \hQ_{m}) 
      - Q_1Q_2\ldots Q_m} \le e^{-m},
  \label{eq:pprime}
  \end{align}
  where we used the fact that $\norm{Q_{\mathrm{rest}}}=1$. 
  Finally, for $\ket{\psi}\in G_\perp$ we get, using the triangle
  inequality and \Eq{eq:pprime},
  \begin{align*}
    \norm{K(m,t,1)\ket{\psi}} 
      \le \norm{\DL(t)\ket{\psi}} 
        + e^{-m}
        \le 2e^{-t \sqrt{\gamma}/25} + e^{-m},
  \end{align*}
  where we used \Lem{lem:DL}. Squaring both sides completes
  the proof.
\end{proof}

\section{Schmidt Rank of the AGSP}
\label{sec:schmidt}

In this section we bound the Schmidt rank of the operator
$K(m,t,\ell)$ across a vertical cut $(c,c+1)$.  Let us begin by
introducing some additional terminology. In the following we shall
use the notation $\SR(O)$ to denote the Schmidt rank of an operator
$O$ across the vertical cut $(c,c+1)$. 

For each coarse-grained projector $Q_i$ with $i\in [m]$, there is a collection of $2t$ column indices $j\in [n]$, such that
column $j$ is in the support of $Q_i$ and every other coarse-grained
projector $Q_{i'}$ (with $i'\neq i$) acts trivially on the $L$ qudits
$j\times [L]$ in the column, see \Fig{fig:agsp}.  Let these columns
be $\Ind_i\subset [n]$ and define $\Ind\EqDef \cup_{i=1}^m\Ind_i$,
so that
\begin{align}
  |\Ind|=2mt.
\label{eq:ind}
\end{align}
The set $\Ind$ is depicted in blue in \Fig{fig:agsp}. 

For each $i=1,2,\ldots, m$ we also define the set of indices
$\Loc_i\subset [n-1]$ such that
\begin{align}
  h_i=\sum_{j\in \Loc_i} H_j
\label{eq:hiloc}
\end{align}
and let $\Loc\EqDef \cup_{i=1}^{m} \Loc_i$. Recall that $Q_i$ is the
projector onto the ground space of the Hamiltonian $h_i$. Note that
$|\Loc_i|=4t-1$ and therefore
\begin{align}
  |\Loc|\leq 4tm.
\label{eq:sizeloc}
\end{align}

Recall that
\begin{align}
  f\EqDef \left\lceil 4\sqrt{tL/\gamma}\right\rceil
\label{eq:f}
\end{align}
is the degree of the polynomial $\mathrm{Step}_{f,g}(x)$ that was
used in the definition \Eq{def:Qhat} of $\widehat{Q}_j$. We shall
also write
\begin{align}
\label{eq:r}
  r\EqDef 11m
\end{align}
for the degree of the polynomial $p_{\mathrm{AND}}$ defined in
\Thm{corrobust}.

Our bound on $\SR\big(K(m,t,\ell)\big)$ is summarized as follows.
\begin{theorem}
  \label{thm:sr}
  Let $c\in [n-1]$ be a column label such that $c \mod
  6t=2t$. Let $\ell,m,t$ be chosen such that
  \begin{align}
   \ell \leq \frac{m^2t^2L}{fr},
  \label{eq:frlbound}
  \end{align}
  where $r,f$ are defined by Eqs. (\ref{eq:f}, \ref{eq:r}).  Then
  the Schmidt rank of $K(m,t,\ell)$ across the cut $(c,c+1)$ is bounded as
  \begin{align}
    \SR\big(K(m,t,\ell)\big)\leq (6mtr)^{3\ell}(6mtdL)^{16mtL}.
  \label{SRmainbound}
  \end{align}
\end{theorem}
In the remainder of this section we prove \Thm{thm:sr}.  We shall
use a variant of the polynomial interpolation technique introduced
in \Ref{AradKLV13}. We introduce a formal complex variable $Z_j$ for
each $j\in \Loc$, and generalize Definition~\ref{def:Qhat} to
\begin{align}
  \widehat{Q}_k (\vec{Z}) 
    \EqDef \mathrm{Step}_{f,g}\left(\frac{1}{4tL}
      \sum_{j\in \Loc_k} Z_j H_j\right) 
        \qquad \qquad k=1,2,\ldots, m.
\label{eq:qz}
\end{align}
Note that $\widehat{Q}_k (1,1,\ldots ,1)=\widehat{Q}_k$ which can be
see from Eqs.~(\ref{def:Qhat},\ref{eq:hiloc},\ref{eq:qz}). We define
\begin{align}
  K(\vec{Z}) \EqDef \left[p_{\mathrm{AND}}
    \left(\widehat{Q}_1(\vec{Z}),\widehat{Q}_2(\vec{Z}),\ldots, 
    \widehat{Q}_m(\vec{Z})\right) Q_{\mathrm{rest}}\right]^\ell ,
\label{eq:wsymz}
\end{align}
where for brevity, we have suppressed the dependence of $K(\vec{Z})$
on $m,t,\ell$. The operator $K(\vec{Z})$ coincides with
$K(m,t,\ell)$ when $\vec{Z}=(1,1,\ldots,1)$, and therefore by upper
bounding $\SR\big(K(\vec{Z})\big)$ for general $\vec{Z}$, we also
upper bound $\SR\big(K(m,t,\ell)\big)$. To this end, we use Eqs.~(\ref{eq:qz},~\ref{eq:wsymz}) and the fact that
$\mathrm{Step}_{f,g}$ is a polynomial to expand $K(\vec{Z})$ as a
multinomial in complex variables with operator coefficients:
\begin{align}
\label{eq:opexpansion}
  K(\vec{Z}) = \sum_{\vec{\beta} = 
    \{\beta_j\}_{j\in \Loc}} \op(\vec{\beta})
    \prod_{j\in \Loc}\br{Z_j}^{\beta_j},
\end{align}
where $\beta_j\in \{0,1,2,\ldots\}$ counts the number of times $H_j$
appears in the operator $\op(\vec{\beta})$. The following lemma
upper bounds $\SR\big(K(\vec{Z})\big)$ in terms of the maximum
Schmidt rank of one of the operators appearing on the right-hand
side of \Eq{eq:opexpansion}.
\begin{lemma}
\label{lem:SR2a}
  \begin{align}
    \SR(K(\vec{Z})) \leq M\cdot \max_{\vec{\beta}} \; 
      \SR\left(\op(\vec{\beta})\right) ,
  \label{eq:M1}
  \end{align}
  where
  \begin{align}
  \label{eq:M}
    M\EqDef \left(3+\frac{3rf\ell}{4tm}\right)^{4tm}.
  \end{align}
\end{lemma}
\begin{proof}
  It suffices to show that the number of nonzero terms on the RHS of
  the expansion \eqref{eq:opexpansion} is at most $M$. Recall that
  each $\widehat{Q}_j$ is a polynomial of degree $f$ and $p_{AND}$
  is a polynomial of degree $r$. Therefore, the operator
  $p_{\mathrm{AND}}\left(\widehat{Q}_1(\vec{Z}),\widehat{Q}_2(\vec{Z}),
  \ldots, \widehat{Q}_m(\vec{Z})\right)$ has total degree
  of at most $rf$ in the $\vec{Z}$ variables and by
  definition \Eq{eq:wsymz}, the total degree of $K(\vec{Z})$ is at
  most $fr\ell$. Comparing with \Eq{eq:opexpansion} we see that
  \begin{align}
    \sum_{j\in \Loc} \beta_j \leq rf\ell
  \label{eq:degreeupperbound}
  \end{align}
  for any tuple $\vec{\beta}$ that appears on the right-hand side of
  \Eq{eq:opexpansion}.  Therefore, the number of nonzero terms in
  the expansion \eqref{eq:opexpansion} is upper bounded by the
  number of tuples of non-negative integers $\{\beta_j\}_{j\in
  \Loc}$ satisfying \Eq{eq:degreeupperbound}. 

  The number of such tuples is\footnote{Recall from elementary
  combinatorics that the number of $p$-tuples of non-negative
  integers $(c_1, c_2, \ldots, c_p)$ such that $\sum_{j=1}^{p} c_j
  \leq q$ is equal to ${p+q \choose p}$.} 
  \begin{align*}
    {|\Loc|+fr\ell \choose |\Loc|} \leq {4mt+fr\ell \choose 4mt} 
      \leq \left(e\cdot \frac{4mt+fr\ell}{4mt}\right)^{4mt}\leq M,
  \end{align*}
  which completes the proof. Here in the first inequality we used
  $|\Loc|\leq 4mt$ and in the second we used the fact that ${ a
  \choose b}\leq (e\cdot a/b)^b$.
\end{proof}

The natural next step is to upper bound
$\SR\big(\op(\vec{\beta})\big)$ for any $\vec{\beta}$ appearing in
\Eq{eq:opexpansion}. Note that $\op(\vec{\beta})$ can be expressed
as a linear combination of products of the operators taken from the
set $\{H_j: j\in \Loc\}\cup \{Q_{\mathrm{rest}}\}$.  For example, it
may contain the product $H_5 H_1 Q_{\mathrm{rest}} H_2
Q_{\mathrm{rest}} H_1 \cdots$. By definition, any such product only
appears in $\op(\vec{\beta})$ if the number of occurrence of $H_j$
is equal to $\beta_j$, and the number of occurrence of
$Q_{\mathrm{rest}}$ is equal to $\ell$. Equipped with this expansion
of $\op(\vec{\beta})$, we can try to upper bound its Schmidt rank by
the number of terms in the expansion multiplied by the maximum
Schmidt rank of any term. Unfortunately, this strategy does not
provide a useful upper bound on $\SR\big(\op(\vec{\beta})\big)$
because the number of terms in the expansion is too large.

Instead of expressing $\op(\vec{\beta})$ as a linear combination of
products of operators from the set $\{H_j: j\in \Loc\}\cup
\{Q_{\mathrm{rest}}\}$, we will show that $\op(\vec{\beta})$ can be
written as a linear combination of a relatively small number of
well-structured operators of a certain form described below. For
each of these well-structured operators there is a column label $k$
(which is close to $c$) such that the Schmidt rank of the operator
across the vertical cut $(k,k+1)$ is small. We will see that this in
turn implies a small Schmidt rank for $\op(\vec{\beta})$ across the
cut $(c,c+1)$ of interest.
 
For any $k\in \Loc$ and positive number $R$, we define the
aforementioned well-structured operators as follows:
\begin{align}
  K_{k}^{\leq R}(\vec{Z}) 
    \EqDef \sum_{\{\beta_j\}_{j\in \Loc}: 
      \beta_k\leq R}\op(\vec{\beta})
        \prod_{j\in \Loc}\br{Z_j}^{\beta_j},
\label{eq:Wkr}
\end{align}
which consists of all the terms in \Eq{eq:opexpansion} satisfying
the additional constraint $\beta_k \leq R$. The following lemma
shows how the Schmidt rank of $\op(\vec{\beta})$ is related to that
of one of these well-structured operators.
\begin{lemma}
\label{lem:SR2}
  Let
  \begin{align*}
    R\EqDef \frac{fr\ell}{2mt}.
  \end{align*}
  For any $\op(\vec{\beta'})$ in the
  expansion~\eqref{eq:opexpansion} there exists a column label $k\in
  \Ind$ and a complex vector $\vec{X}=\{X_j\}_{j\in \Loc}$ such that
  \begin{align}
  \label{eq:M2}
    \SR\big(\op(\vec{\beta'})\big)
      \le M \cdot \SR\left(K_{k}^{\leq R} (\vec{X})\right) ,
  \end{align}
 where $M$ is defined in \Eq{eq:M}.
\end{lemma}

\begin{proof}
  Consider any operator $\op(\vec{\beta'})$ appearing in
  \Eq{eq:opexpansion} and recall that $\sum_j \beta'_j\le rf\ell$.
  By \Eq{eq:ind}, the subset of column labels
  $\Ind\subset \Loc$ has size $|\Ind|=2mt$ and therefore
  \begin{align*}
    \sum_{j\in \Ind} \beta'_j 
      \leq \sum_{j\in \Loc}\beta'_j \leq rf\ell.
  \end{align*}
  It follows that there must exist some column label $k\in \Ind$
  such that
  \begin{align}
    \beta'_k \leq \frac{r f \ell}{|\Ind|}=\frac{rf\ell}{2mt} = R.
  \label{eq:betaprime}
  \end{align}
  So let $k$ be fixed to the column label satisfying the above, and
  consider the operator $K_k^{\leq R}(\vec{Z})$ defined in
  \Eq{eq:Wkr}. Note that since the tuple $\vec{\beta}'$ satisfies
  \Eq{eq:betaprime}, it appears as one of the terms in the sum in
  \Eq{eq:Wkr}.  We have the following:

  \begin{claim}
  \label{clm:invertmat} 
    There exists a collection of complex tuples $\vec{X}^{(1)},
    \vec{X}^{(2)}, \ldots, \vec{X}^{(M)}$ such that
    $\op(\vec{\beta}')$ is a linear combination (with complex
    coefficients) of the operators 
    \begin{align}
      K^{\leq R}_k(\vec{X}^{(1)}), K^{\leq R}_k(\vec{X}^{(2)}), 
        \ldots, K^{\leq R}_k(\vec{X}^{(M)}).
    \label{eq:lincomb}
    \end{align}
  \end{claim}

  \begin{proof}
    Let $\cT$ be the set of all tuples of nonnegative integers
    $\vec{\beta}=\{\beta_j\}_{j\in \Loc}$ such that 
    \Eq{eq:degreeupperbound} is satisfied and $\beta_k\leq R$. That
    is, 
    \begin{align*}
      \cT=\{\vec{\beta} =
        \{\beta_j\}_{j\in \Loc}: \sum_{j\in \Loc} 
          \beta_j\leq rf\ell \qquad \text{and } \beta_k\leq R \}.
    \end{align*}
    The set $\cT$ has size upper bounded as $|\cT|\leq M$ where $M$ is
    given by \Eq{eq:M}. Consider the following system of equations.
    \begin{align*}
      &K^{\leq R}_k(\vec{X}^{(1)})
        = \sum_{\vec{\beta}\in \cT} \op(\vec{\beta})
            \prod_{j\in \Loc}\br{X^{(1)}_j}^{\beta_j},\\
      &K^{\leq R}_k(\vec{X}^{(2)})
        = \sum_{\vec{\beta}\in \cT} \op(\vec{\beta})
            \prod_{j\in \Loc}\br{X^{(2)}_j}^{\beta_j},\\
      &\vdots \\
      &K^{\leq R}_k(\vec{X}^{(|\cT|)})
        = \sum_{\vec{\beta}\in \cT} \op(\vec{\beta})
            \prod_{j\in \Loc}\br{X^{(|\cT|)}_j}^{\beta_j} .
    \end{align*}
    We now show that there exists at least one choice of
    $\vec{X}^{(1)}, \ldots, \vec{X}^{(|\cT|)}$ such that this system
    of equations can be inverted to obtain $\op(\vec{\beta})$ as a
    linear combination of the operators appearing on the
    left-hand-side, for any $\vec{\beta}\in \cT$. This is sufficient
    to complete the proof, as $|\cT|\leq M$. 

    Consider the (square) matrix
    \begin{align*}
      G_{\alpha, \vec{\beta}}\EqDef 
        \prod_{j\in \Loc}\br{X^{(\alpha)}_j}^{\beta_j} 
          \qquad \alpha=1,2,\ldots, |\cT| \qquad \vec{\beta}\in \cT
    \end{align*}
    We show that this matrix has non-zero determinant for some
    choice of $\vec{X}^{(1)}, \ldots, \vec{X}^{(|\cT|)}$q. This
    implies the matrix is invertible and completes the proof. 

    Fix some order over $\vec{\beta}\in \cT$ and let
    $\vec{\beta}(\alpha)$ be the $\alpha$-th $\vec{\beta}$ in this
    order. We have
    \begin{align*}
      \det(G) = \sum_{\pi}(-1)^{sign(\pi)}
        \prod_{\alpha}G_{\alpha, \vec{\beta}(\pi(\alpha))} = \sum_{\pi}(-1)^{sign(\pi)}
        \prod_{\alpha}\br{\prod_{j\in \Loc}
          \br{X^{(\alpha)}_j}^{\beta(\pi(\alpha))_j}},  
    \end{align*}
    where $\pi$ is a permutation over the set $[1:|\cT|]$. We would
    like to show that there exists at least one choice of
    $\vec{X}^{(1)}, \ldots, \vec{X}^{(|\cT|)}$ such $\det(G)\neq 0$.
    To that aim, we consider $\det(G)$ as a multinomial over the
    variables $\vec{X}^{(1)}, \ldots, \vec{X}^{(|\cT|)}$ and show
    that it is not identically zero. Indeed, as the tuples
    $\vec{\beta}(\alpha)$ are distinct for different $\alpha$'s, it
    follows that for any two distinct permutations $\pi_1, \pi_2$,
    the multinomials $\prod_{\alpha}\br{\prod_{j\in
    \Loc}\br{X^{(\alpha)}_j}^{\beta(\pi_1(\alpha))_j}}$ and
    $\prod_{\alpha}\br{\prod_{j\in
    \Loc}\br{X^{(\alpha)}_j}^{\beta(\pi_2(\alpha))_j}}$ are
    distinct.  Therefore, $\det(G)$ is a sum of distinct
    multinomials with coefficient in $\{-1,1\}$, which implies in
    particular that $\text{det}(G)$ is not identically zero.
  \end{proof}
  Now Claim~\ref{clm:invertmat} implies in particular that
  $\op(\vec{\beta}')$ can be expressed as a linear combination of
  $M$ operators \Eq{eq:lincomb} and therefore has Schmidt rank upper
  bounded by $M$ times the maximum Schmidt rank of one of these
  operators. This establishes \Eq{eq:M2} and completes the proof of
  Lemma \ref{lem:SR2}.
\end{proof}
Combining \Lem{lem:SR2a} with \Lem{lem:SR2} we obtain the
following corollary:
\begin{corollary}
\label{cor:SR2}
  Let $R=\frac{fr\ell}{2mt}$. There exists a column label
  $k\in \Ind$ and a complex vector $\vec{X}=\{X_j\}_{j\in \Loc}$
  such that
  \begin{align*}
    \SR\big(K(\vec{Z})\big)
      \le \left(3+\frac{3fr\ell}{4mt}\right)^{8mt} 
        \SR\left(K_{k}^{\leq R} (\vec{X})\right).
  \end{align*}
  for all complex vectors $\vec{Z}$.  
\end{corollary}

The last ingredient we will use to prove \Thm{thm:sr} is a bound on the
Schmidt rank of $K_k^{\le R}(\vec{X})$.
\begin{lemma}
\label{lem:SR3} 
  Let $k\in \Ind$ be a column label, $N$ be a positive integer, and
  $\vec{Z}=\{Z_j\}_{j\in \Loc}$ be a tuple of complex numbers. Then
  \begin{align}
    \SR(K_k^{\leq N} (\vec{Z}))
      \leq 2^{N+\ell} L^N d^{4N+6mtL} r^{3\ell}\left(3
        + \frac{3fr\ell}{2N}\right)^{2\ell+2N}.
  \label{eq:lemSR}
  \end{align}
\end{lemma}

\begin{proof}

  To bound $\SR(K_k^{\leq  N} (\vec{Z}))$, which is defined 
  with respect to the cut $(c,c+1)$, we will first 
  bound the Schmidt rank of $K_k^{\leq N}
  (\vec{Z})$ across the cut $(k,k+1)$, which we write as
  \begin{align}
    \SR_{k,k+1}(K_k^{\leq N} (\vec{Z})).
  \label{eq:srk}
  \end{align}
  Since the column $c$ sits in the middle of the $6mt-2t$ columns
  that support $Q_1, Q_2, \ldots, Q_m$ (see \Fig{fig:agsp}), it
  follows that the distance between $c$ and $k$ must not exceed
  $3mt$. Using the fact that for any
  operator $O$ we have
  \begin{align*}
    \SR_{c,c+1}(O)\leq (d^L)^{2|c-k|}\SR_{k,k+1}(O),
  \end{align*}
  which follows from the fact that the Hilbert space of each column
  has dimension $d^L$, we find that the Schmidt rank across the cut
  $(c,c+1)$
  is bounded as
  \begin{align}
    \SR(K_k^{\leq N} (\vec{Z})) \leq d^{6mtL}
      \SR_{k,k+1}(K_k^{\leq N} (\vec{Z})).
  \label{eq:sric}
  \end{align}

  Let us then proceed with bounding the the Schmidt rank across the 
  $(k,k+1)$ cut. By definition of the set $\Ind$, for the
  given column label $k\in \Ind$, there is a unique $u\in
  \{1,2,\ldots, m\}$ such that $k\in \Ind_u$.  Below, we decompose
  the polynomial operator
  $p_{\mathrm{AND}}\left(\widehat{Q}_1(\vec{Z}),
  \widehat{Q}_2(\vec{Z}),\ldots, \widehat{Q}_m(\vec{Z})\right)$,
  which appears in the definition of $K(\vec{Z})$ in
  \eqref{eq:wsymz}, in powers of $\widehat{Q}_u(\vec{Z})$. Using the
  fact that the operators $\widehat{Q}_1(\vec{Z}),
  \widehat{Q}_2(\vec{Z}),\ldots, \widehat{Q}_m(\vec{Z})$ commute
  with each other, \Thm{corrobust} implies
  \begin{align}
    p_{\mathrm{AND}}\left(\widehat{Q}_1(\vec{Z}),
      \widehat{Q}_2(\vec{Z}),\ldots, \widehat{Q}_m(\vec{Z})\right)
        &= \sum_{\substack{\{i_1, \ldots i_m\}\\ 
            i_1+\ldots +i_m \leq 5m}} 
          A_{i_1}\big(\widehat{Q}_1(\vec{Z})\big)
            \cdot A_{i_2}\big(\widehat{Q}_2(\vec{Z})\big)
              \cdot\ldots\cdot A_{i_m}\big(\widehat{Q}_m(\vec{Z})\big)
            \nonumber\\
    &= \sum_{i_u=0}^{5m}\sum_{s=0}^{5m-i_u}\mathcal{L}(s) A_{i_u}\big(\widehat{Q}_{u}(\vec{Z})\big) \mathcal{R}(s)
  \label{eq:leftright}
  \end{align}
  where 
  \begin{align*}
    \mathcal{L}(s) \EqDef \sum_{i_1+\ldots+ i_{u-1} =s} 
      A_{i_1}\br{\widehat{Q}_1(\vec{Z})}\cdot\ldots\cdot
        A_{i_{u-1}}\br{\widehat{Q}_{u-1}(\vec{Z})}
  \end{align*}
  is supported only on the columns $j<k$ to the left of $k$, and 
  \begin{align*}
    \mathcal{R}(s) \EqDef \sum_{i_{u+1}+\ldots+ i_m \leq 5m-s-i_u} 
      A_{i_{u+1}}\br{\widehat{Q}_{u+1}(\vec{Z})}\cdot\ldots\cdot 
        A_{i_m}\br{\widehat{Q}_m(\vec{Z})}
  \end{align*}
  is supported only on the columns $j\geq k+1$ 
  to the right of $k$. In particular, neither $\mathcal{L}(s)$ nor
  $\mathcal{R}(s)$ increases the Schmidt rank across the cut
  $(k,k+1)$.
  
  Next, recall from \Thm{corrobust} that each $A_i$ is a polynomial
  of degree $2i+1$. Expanding
  $A_{i_u}\big(\widehat{Q}_{u}(\vec{Z})\big)$ in powers of
  $\widehat{Q}_{u}(\vec{Z})$ in \Eq{eq:leftright} and using the fact
  that all operators commute, we see that
  $p_{\mathrm{AND}}\left(\widehat{Q}_1(\vec{Z}),\widehat{Q}_2(\vec{Z}),
  \ldots, \widehat{Q}_m(\vec{Z})\right)$ can be expressed as a
  linear combination of at most
  \begin{align*}
    5m\cdot 5m\cdot (2\cdot 5m+1) \leq (11m)^3=r^3
  \end{align*}
  terms of the form 
  \begin{align*}
    \left(\widehat{Q}_u(\vec{Z})\right)^a \mathcal{L}\;\mathcal{R}
  \end{align*}
  where $a$ is a non-negative integer, the operator $\mathcal{L}$ is
  supported only on columns $j<k$ and the operator $\mathcal{R}$ is
  supported only on columns $j>k+1$. Both $\mathcal{L}$ and
  $\mathcal{R}$ depend on $\vec{Z}$, but neither of them depend on
  the variable $Z_k$ corresponding to column $k$. Therefore 
  \begin{align*}
    K(\vec{Z})=\left(p_{\mathrm{AND}}\left(\widehat{Q}_1(\vec{Z}),
      \widehat{Q}_2(\vec{Z}),\ldots, 
      \widehat{Q}_m(\vec{Z})\right) Q_{\mathrm{rest}}\right)^\ell.
  \end{align*}
  can be expressed as a linear combination of at most $r^{3\ell}$
  terms of the form
  \begin{align}
  \label{eq:wsymterms}
    \left(\widehat{Q}_u(\vec{Z})\right)^{a_1} 
      \mathcal{L}^{(1)}\mathcal{R}^{(1)} Q_{\mathrm{rest}}
      \left(\widehat{Q}_u(\vec{Z})\right)^{a_2} 
      \mathcal{L}^{(2)}\mathcal{R}^{(2)} Q_{\mathrm{rest}}\ldots 
      \left(\widehat{Q}_u(\vec{Z})\right)^{a_\ell} 
      \mathcal{L}^{(\ell)}\mathcal{R}^{(\ell)} Q_{\mathrm{rest}},
  \end{align}
  corresponding to possibly different choices of operators
  $\{\mathcal{L}^{(j)}, \mathcal{R}^{(j)}\}$ and powers $\{a_j\}$
  satisfying
  \begin{align}
    \sum_{q=0}^{\ell} a_q\leq r\ell.
  \label{eq:sumofas}
  \end{align}
  By expanding each of the polynomials $\widehat{Q}_u(\vec{Z})$ we
  may expand each term \Eq{eq:wsymterms} as a polynomial in $Z_k$
  with operator coefficients. We are interested in $K_{k}^{\leq
  N}(\vec{Z})$ which includes only those terms with at most $N$
  powers of $Z_k$ (see the definition in \Eq{eq:Wkr}).
  In the following we fix a term \Eq{eq:wsymterms} (i.e., a choice
  of $\{\mathcal{L}^{(j)},\mathcal{R}^{(j)}\}$ and $\{a_j\}$) and
  bound the Schmidt rank of all such operators with at most $N$
  powers of $Z_k$ arising from it. Then we multiply by $r^{3\ell}$
  to obtain the desired upper bound on the Schmidt rank of
  $K_{k}^{\leq N}(\vec{Z})$.

  So let us fix a term \Eq{eq:wsymterms}. Now, $\hQ_u(\vec{Z})$ is a
  polynomial of degree $f$ in the subregion operator
  \begin{align*}
    h_u(\vec{Z})=\sum_{j\in \Loc_{u}}H_jZ_j.
  \end{align*}
  Let 
  \begin{align*}
    C\EqDef \sum_{j\in \Loc_{u}: j<k}H_jZ_j 
      \qquad \text{ and }\qquad 
      D\EqDef \sum_{j\in \Loc_{u}: j>k}H_jZ_j,
  \end{align*}
  so that $h_u(\vec{Z})=C+H_kZ_k+D$.  Since $[C,D]=0$, each
  degree-$(a_qf)$ polynomial $\left(\hQ_u(\vec{Z})\right)^{a_q}$
  appearing in \Eq{eq:wsymterms} is a linear combination of terms of
  the form
  \begin{align*}
    \br{C^{\alpha^{(0)}_q}D^{\beta^{(0)}_q}}H_{k}Z_k
      \br{C^{\alpha^{(1)}_q}D^{\beta^{(1)}_q}} H_{k}Z_k
        \ldots \br{C^{\alpha^{(T_q-1)}_q}
        D^{\beta^{(T_q-1)}_q}}H_kZ_k \br{C^{\alpha^{(T_q)}_q}
        D^{\beta^{(T_q)}_q}}
  \end{align*}
  where $0\leq T_q\leq a_qf$, and the nonnegative integers
  $\{\alpha^{(j)}_q, \beta^{(j)}_q\}$ satisfy $\sum_{j=0}^{T_q}
  \left(\alpha^{(j)}_q+\beta^{(j)}_q\right) \leq a_qf$.
  Equation~\eqref{eq:wsymterms} then expands into terms of the form
  \begin{align}
  \label{eq:SRreduct4}
    &\br{C^{\alpha^{(0)}_1}D^{\beta^{(0)}_1}}H_{k}Z_k
      \br{C^{\alpha^{(1)}_1}D^{\beta^{(1)}_1}}
      H_{k}Z_k\ldots \br{C^{\alpha^{(T_1)}_1}D^{\beta^{(T_1)}_1}} \br{\mathcal{L}^{(1)} \mathcal{R}^{(1)}}
        Q_{\mathrm{rest}}\cdot \nonumber\\
    &\br{C^{\alpha^{(0)}_2}D^{\beta^{(0)}_2}}H_{k}Z_k
      \br{C^{\alpha^{(1)}_2}D^{\beta^{(1)}_2}}
      H_{k}Z_k\ldots  \br{C^{\alpha^{(T_2)}_2}D^{\beta^{(T_2)}_2}}\br{\mathcal{L}^{(2)}\mathcal{R}^{(2)}}
        Q_{\mathrm{rest}}\cdot\ldots\nonumber\\ 
    &\br{C^{\alpha^{(0)}_\ell}D^{\beta^{(0)}_\ell}}H_{k}Z_k
      \br{C^{\alpha^{(1)}_\ell}D^{\beta^{(1)}_\ell}}
      H_{k}Z_k\ldots  \br{C^{\alpha^{(T_\ell)}_\ell}D^{\beta^{(T_\ell)}_\ell}}\br{\mathcal{L}^{(\ell)}\mathcal{R}^{(\ell)}}Q_{\mathrm{rest}}.
  \end{align}
  where
  \begin{align}
    \sum_{q=1}^{\ell} \sum_{j=0}^{T_{q}} 
      \left(\alpha^{(j)}_q+\beta^{(j)}_q\right)
        \leq \sum_{q=1}^{\ell} a_q f \leq f r \ell,
  \label{eq:sumalpha}
  \end{align}
  and in the second inequality we used \Eq{eq:sumofas}.  Since we
  are concerned with $K_{k}^{\leq N}(\vec{X})$, we only consider the
  terms of the form \eqref{eq:SRreduct4}, in which $H_kZ_k$ occurs
  at most $N$ times, that is,
  \begin{align}
    \sum_{q=1}^{\ell} T_q \leq N.
  \label{eq:Tconstraint}
  \end{align}
  Let us now count the number of such terms that satisfy the
  constraints Eqs.~(\ref{eq:sumalpha}, \ref{eq:Tconstraint}). There
  are ${N+\ell \choose \ell}\leq 2^{N+\ell}$ tuples $(T_1,\ldots,
  T_\ell)$ of nonnegative integers satisfying \Eq{eq:Tconstraint}.
  For a fixed tuple $(T_1,\ldots, T_\ell)$, note that the left-hand
  side of \Eq{eq:sumalpha} is a sum of at most
  $$2\sum_{q=1}^{\ell}(T_q+1)=2(N+\ell)$$
  nonnegative integers. Thus, for each tuple $(T_1,\ldots, T_\ell)$,
  the number of choices for these nonnegative integers
  $\{\alpha^{(j)}_q, \beta^{(j)}_q\}$ satisfying \Eq{eq:sumalpha} is
  at most
  \begin{align*}
    { 2(N+\ell)+fr\ell \choose 2(N+\ell)}
      \leq \left(e\cdot 
        \frac{2(N+\ell)+fr\ell}{2(N+\ell)}\right)^{2N+2\ell}
      \leq \left(3+\frac{3fr\ell}{2N}\right)^{2N+2\ell} .
  \end{align*}
  Each choice for $(T_1,\ldots, T_\ell)$ and $\{\alpha^{(j)}_q,
  \beta^{(j)}_q\}$ corresponds to exactly one operator as given in
  \Eq{eq:SRreduct4}. Note that the operator $H_k$ is a sum of at
  most $L$ projectors $P_{ij}$ which each have Schmidt rank at most
  $d^4$ across the cut $k,k+1$. Therefore the operator
  \Eq{eq:SRreduct4} has Schmidt rank at most $(Ld^4)^{N}$ across the
  cut $(k,k+1)$, as the term $H_k$ occurs at most $N$ times, and the
  operators `$\mathcal{L},\mathcal{R},C,D$' and $Q_{\mathrm{rest}}$
  do not increase the Schmidt rank. Collecting all the contributions
  to the Schmidt rank across the cut $(k,k+1)$, we find that
  \begin{align*}
    \SR_{k,k+1}(K_k^{\leq N} (\vec{Z}))
      \leq \underbrace{\left(r^{3\ell}\right)}_{\text{\# 
        of terms \Eq{eq:wsymterms}}} \cdot \quad 
     \underbrace{2^{N+\ell}\left(3+\frac{3fr\ell}{2N}
       \right)^{2N+2\ell}}_{\substack{\text{\# of 
        operators \Eq{eq:SRreduct4}}\\ 
       \text{ with $\leq N$ powers of $Z_k$}\\ 
       \text{arising from each term \Eq{eq:wsymterms}}}}
  \quad \cdot \underbrace{(Ld^4)^{N}.}_{\substack{\text{SR 
    of each op. \Eq{eq:SRreduct4} }\\ 
      \text{with $\leq N$ powers of $Z_k$}}}
  \end{align*}
   Plugging this into \Eq{eq:sric} we obtain the desired bound
   \Eq{eq:lemSR} on the Schmidt rank across the cut $(c,c+1)$.
\end{proof}

\begin{proof}[Proof of \Thm{thm:sr}]
  Combining Corollary~\ref{cor:SR2} and \Lem{lem:SR3} with
  $N=R=\frac{fr\ell}{2mt}$, we see that
  \begin{align}
    \SR(K(m,t,\ell)) &\leq 
      \left[ \left(3+\frac{3fr\ell}{4mt}\right)^{8mt}\right]
      \left[2^{\ell+\frac{fr\ell}{2mt}} L^{\frac{fr\ell}{2mt}}
        d^{\frac{2fr\ell}{mt}+6mtL}r^{3\ell}
      \left(3+3mt\right)^{2\ell +\frac{fr\ell}{mt}}\right] .
  \label{eq:3lems}
  \end{align}
  Here the two terms in square parentheses come from the
  corollary and lemma, respectively. Using \Eq{eq:frlbound} we see that
  \begin{align}
    2^{\ell}\left(3+3mt\right)^{2\ell+\frac{fr\ell}{mt}}
      d^{\frac{2fr\ell}{mt}+6mtL}L^{\frac{fr\ell}{2mt}} 
        & \leq 2^{\ell}(6mt)^{2\ell+mtL} d^{8mtL}L^{mtL/2}\nonumber\\
    &\leq (6mt)^{3\ell+mtL}(dL)^{8mtL} ,
  \label{eq:s2}
  \end{align}
  and using \Eq{eq:frlbound} again we get 
  \begin{align}
  \label{eq:s3}
    2^{\frac{fr\ell}{2mt}}\left(3+\frac{3fr\ell}{4mt}\right)^{8mt}
      \leq 2^{mtL/2}\left(3+\frac{3mtL}{4}\right)^{8mt}
      \leq \left(2^{1/16}\cdot 4\cdot mtL\right)^{8mtL} 
      \leq \left(6mtL\right)^{8mtL}.
  \end{align}
  Plugging the bounds Eqs.~(\ref{eq:s2},\ref{eq:s3}) into
  \Eq{eq:3lems}, we arrive at \Eq{SRmainbound} and complete the
  proof.
\end{proof}

\section{Proof of the subvolume law for a vertical cut}
\label{sec:arealaw1D}

We now prove \Thm{thm:subv_cut}.
\begin{proof}
  Let us begin by specifying choices for the positive integers $t,
  \ell$ and odd positive integer $m$ which determine the AGSP
  $K(m,t,\ell)$. We choose the coarse-graining parameter as follows:
  \begin{align}
    t=\left\lceil \frac{25m}{\sqrt{\gamma}}\right\rceil.
  \label{eq:tchoice}
  \end{align}
  With this choice, the bound on the shrinking factor $\Delta$ of
  $K(m,t,\ell)$ from \Eq{eq:Kpart1} can be simplified to
  \begin{align}
    \Delta\leq \Delta' \EqDef 3^{2\ell}e^{-2m\ell}.
  \label{eq:shrsimplified}
  \end{align}
  For future reference we note that since $m$ is a positive integer
  and $\gamma\leq 1$ we have
  \begin{align}
    t\leq \frac{26m}{\sqrt{\gamma}}.
    \label{eq:tupperbound}
  \end{align}
  We choose 
  \begin{align}
    \ell=\left \lfloor \frac{m^2t^2L}{fr}\right\rfloor
  \label{eq:ellchoice}
  \end{align}
  so that the condition \Eq{eq:frlbound} is satisfied. For future
  reference we note that
  \begin{align}
    (\ell+1)\geq \frac{m^2t^2L}{fr}
      = \frac{mt^2L}{11f}\geq \frac{mt^{3/2}\sqrt{L\gamma}}{55}
      \geq \frac{(25)^{3/2}}{55}
        \cdot \frac{m^{5/2}\sqrt{L}}{\gamma^{1/4}},
  \end{align} 
  where the second inequality uses
  $f= \lceil 4\sqrt{\frac{tL}{\gamma}}\rceil \leq
  5\sqrt{\frac{tL}{\gamma}}$. Since $25^{3/2}/55>2$ and
  $m,L,\gamma^{-1}\geq 1$ we see that
  \begin{align}
    \ell\geq \frac{2m^{5/2}\sqrt{L}}{\gamma^{1/4}}-1 
      \geq \frac{m^{5/2}\sqrt{L}}{\gamma^{1/4}}.
  \label{eq:ell_lowerb}
  \end{align}

  It remains to choose $m$. Let us choose it to ensure that the
  parameters of the AGSP $K(m,t,\ell)$ satisfy $D\cdot \Delta\leq
  1/2$ so that \Thm{thm:AGSParealaw} can be applied. 

  Here $D$ is the upper bound on $\mathrm{SR}(K(m,t,\ell))$ given by \Thm{thm:sr}
  and $\Delta$ is upper bounded in \Eq{eq:shrsimplified}. Using
  these bounds, plugging in $r=11m$, and taking logs we see that
  \begin{align}
    &D\cdot \Delta\leq \frac{1}{2} \quad \text{ if the following 
      condition holds:} \nonumber\\ 
    & 3\ell\log(66m^2t)+16mtL\log(6mtdL)-2m\ell+2\ell\log(3)
      \leq -\log(2) . \label{eq:ddconstraint}
  \end{align}
  We now choose
  \begin{align}
    m \EqDef \left \lceil \frac{10^4 L^{1/3}}{\gamma^{1/6}} 
      \log^{2/3}(dL\gamma^{-1})\right\rceil_{\mathrm{Odd}} ,
  \label{eq:mchoice}
  \end{align}
  where $\left \lceil x \right \rceil_{\mathrm{Odd}}$ denotes the
  smallest odd integer which is at least $x$ (recall that in the
  definition of $K(m,t,\ell)$, we require $m$ to be an odd positive
  integer). Note that since $\gamma\leq 1$, $L\geq 1$, and $d\geq 1$
  we have
  \begin{align}
    10^4\leq  m\leq \frac{2\cdot 10^4 L^{1/3}}{\gamma^{1/6}} 
      \log^{2/3}(dL\gamma^{-1}) .
  \label{eq:mbounds}
  \end{align}
  
  \begin{claim}  
  \label{claim:choices} 
    The chosen parameters $m,t,\ell$ given by
    Eqs.~(\ref{eq:mchoice}, \ref{eq:tchoice}, \ref{eq:ellchoice})
    satisfy the inequality \Eq{eq:ddconstraint}.
  \end{claim}
  
  The proof of the Claim is provided below. Let us now see how it
  implies the theorem. First consider the special case where the cut
  $(c,c+1)$ satisfies $c\mod 6t=2t$.  Since we have $D\cdot \Delta
  \leq D\cdot \Delta'\leq \frac{1}{2}$ we may apply
  \Thm{thm:AGSParealaw} which states that the entanglement entropy
  of $\ket{\Omega}$ across the cut $(c,c+1)$ is upper bounded by
  \begin{align}
    10\log(D) \leq 10\log\left(\frac{1}{2\Delta'}\right)
      = 10\log\left(\frac{3^{-2\ell}}{2} e^{2m\ell}\right)
      \leq 20m\ell.
  \label{eq:eec}
  \end{align}
  Now substituting $f\geq 4\sqrt{tL/\gamma}$, $r=11m$, and
  \Eq{eq:tupperbound} in \Eq{eq:ellchoice} gives
  \begin{align}
    \ell \leq \frac{m\sqrt{L} t^{3/2}}{44} \sqrt{\gamma}
      \leq \frac{26^{3/2}}{44} \frac{m^{5/2} \sqrt{L}}{\gamma^{1/4}}
      \leq \frac{4m^{5/2}\sqrt{L}}{\gamma^{1/4}}.
  \label{eq:lupperbound}
  \end{align}
  Plugging \Eq{eq:lupperbound} into \Eq{eq:eec} and using
  \Eq{eq:mbounds} gives 
  \begin{align}
    10\log(D)\leq 80\cdot \frac{m^{7/2}\sqrt{L}}{\gamma^{1/4}}
      \leq \frac{91\cdot 10^{15} L^{5/3}}{\gamma^{5/6}} 
      \log^{7/3}(dL\gamma^{-1}).
  \label{eq:Dterm}
  \end{align}
  This completes the proof of the theorem in the special case where
  $c\mod 6t=2t$. If $c\mod 6t \neq 2t$ then we find the nearest $c$
  that satisfies this property, losing an entanglement entropy of
  $3tL\log d$, by the subadditivity of entropy. Note that
  \begin{align}
    3tL\log(d)\leq \frac{78mL}{\sqrt{\gamma}}\log(d)
      \leq \frac{156\cdot 10^4 L^{4/3}}{\gamma^{2/3}}
        \log^{5/3}(dL\gamma^{-1}) 
      \leq \frac{156\cdot 10^4L^{5/3}}{\gamma^{5/6}}
        \log^{7/3}(dL\gamma^{-1})
  \label{eq:subadd}
  \end{align}
  where in the first inequality we used \Eq{eq:tupperbound}, in the
  second one we used the upper bound \Eq{eq:mbounds} and the fact
  that $\log(d)\leq \log(dL\gamma^{-1})$, and in the third
  inequality we used the facts that $\gamma\leq 1$ and $L\geq 1$.
  The entanglement entropy across the cut of interest is then at
  most
  \begin{align}
    3tL\log(d)+ 10\log(D) 
      \leq \frac{10^{17}L^{5/3}}{\gamma^{5/6}}
        \log^{7/3}(dL\gamma^{-1}),
  \label{eq:lastinequality}
  \end{align}
  where we used Eqs.~(\ref{eq:Dterm}, \ref{eq:subadd}), completing
  the proof.
\end{proof}

\begin{proof}[Proof of Claim \ref{claim:choices}]
  Note that for any $m\geq 2$ (Cf. \Eq{eq:mbounds}) and any
  $\ell\geq 1$ we have
  \begin{align}
    -2m\ell+2\ell\log(3)+\log(2) \leq -m\ell.
  \label{eq:m}
  \end{align}
  Thus it remains to show that 
  \begin{align}
    3\ell \cdot \log(66m^2t)+16mtL\log(6mtdL)-m\ell\leq 0.
  \label{eq:verify}
  \end{align}
  Below we show that
  \begin{align}
    3\ell\log(66m^2t)\leq \frac{m\ell}{2} 
      \qquad \text {and} \qquad 16mtL\log(6mtdL)
        \leq \frac{m\ell}{2},
  \label{eq:ineq}
  \end{align}
  from which~\eqref{eq:verify} follows directly.

  It remains to establish \Eq{eq:ineq}. The first part follows using
  \Eq{eq:tupperbound} and \Eq{eq:mbounds} which give
  \begin{align}
    3\log(66m^2t)\leq 3\log(1716m^3\gamma^{-1/2})\leq 3\log(1716 m^6)\leq m/2
  \end{align}
  where in the second inequality we used the fact that
  $\gamma^{-1/2}\leq m^3$ and in the third inequality we used the
  fact that $3\log(1716m^6)\leq m/2$ for $m\geq 10^4$. The fact that $\gamma^{-1/2}\leq m^3$ follows from our definition
  of $m$ in \eqref{eq:mchoice}, which implies $m\le \frac{10^4
  L^{1/3}}{\gamma^{1/6}} \log^{2/3}(dL\gamma^{-1})\le
  \gamma^{-1/6}$.
  To establish
  the second part of \Eq{eq:ineq}, we use \Eq{eq:ell_lowerb} and
  \Eq{eq:mchoice} to get
  \begin{align}
    \frac{m\ell}{2}\geq \frac{m^{7/2}\sqrt{L}}{2\gamma^{1/4}}
      \geq \frac{(10^{4})^{7/2}}{2} 
        \cdot\frac{L^{5/3}}{\gamma^{5/6}} \log^{7/3}(dL\gamma^{-1}).
  \label{eq:ml}
  \end{align}
  Also note, using Eqs.~(\ref{eq:tupperbound},\ref{eq:mbounds}),
  that
  \begin{align}
    mtL\leq 26m^2L\gamma^{-1/2}
      \leq \frac{104\cdot 10^{8}L^{5/3}}{\gamma^{5/6}} 
        \log^{4/3}(dL\gamma^{-1})
  \label{eq:mtl}
  \end{align}
  and therefore
  \begin{align}
    16mtL\log(6mtdL)&\leq 16mtL
      \left( \log\br{\frac{dL^{5/3}}{\gamma^{5/6}}}
        + \log(6\cdot 104\cdot 10^{8})
        + \frac{4}{3}\log(\log(dL\gamma^{-1}))\right)\nonumber\\
     &\leq 16mtL \left(3\log(dL\gamma^{-1})+11
      + \frac{4}{3}\log(dL\gamma^{-1})\right)\nonumber\\
     &\leq (16\cdot 16)mtL \log(dL\gamma^{-1})\nonumber\\
     &\leq \frac{256\cdot 104\cdot 10^8L^{5/3}}{\gamma^{5/6}}
       \log^{7/3}(dL\gamma^{-1}),
  \label{eq:mt}
  \end{align}
  where in the first and last steps we used \Eq{eq:mtl}. Combining
  Eqs.~(\ref{eq:ml}, \ref{eq:mt}) and using the fact that
  $256\cdot104\cdot10^8 < 10^{14}/2$ establishes the second part of
  \Eq{eq:ineq} and completes the proof.
\end{proof}

\section{Subvolume law for rectangular regions}
\label{sec:arealawrectangle}

In this Section we consider bipartitions of the 2D grid into a
rectangular region and its complement (see \Fig{fig:cut} (b)) and
prove \Thm{thm:subr_cut}. Since the proof closely follows that of
\Thm{thm:subv_cut}, we shall describe the (minor) modifications
needed. 

The main observation that we will need is that the construction of
the AGSP $K(m,t,\ell)$ and the proof of \Thm{thm:subv_cut} are
essentially one-dimensional, as they are entirely based upon the
expression \Eq{eq:twoloc} for the Hamiltonian as a 1D
nearest-neighbor chain of columns.  In particular, we may reproduce
the proofs and definitions in Sections
\ref{sec:agsp}-\ref{sec:arealaw1D} to bound the entanglement entropy
for any bipartition of the 2D grid with the following properties:
\begin{description}
  \item [1.] We can partition the qudits of the 2D grid into 
    subsets $S_1,S_2, \ldots, $ such that the Hamiltonian takes the
    form $H=\sum_{i} H_i$, where $H_i$ is a sum of projectors which
    act nontrivially only on subsets $S_i$ and $S_{i+1}$.
    
  \item [2.] The positive integer $L$ is an upper bound on the number 
    of qudits in each subset, and on the number of projectors in
    each nearest-neighbor term $H_i$.
    
  \item [3.] The bipartition of interest corresponds to a bipartition
    separating subsets $S_i$ with $i\leq c$ from those with $i\geq
    c+1$.
\end{description}
Under these conditions we obtain an upper bound
$\frac{CL^{5/3}}{\gamma^{5/6}} \log^{7/3}(dL\gamma^{-1})$ on the
entanglement entropy of the ground state, for some universal
constant $C>0$. 

Looking more closely, note that \textbf{1., 3.} allow us to define
the coarse grained projectors and AGSP as in \Sec{sec:agsp} and the
proof then only requires the following slightly weaker version of
condition \textbf{2.} which concerns only a region of $O(mt)$ 
subsets $\{S_i\}$ of the qudits centered around the cut.

\begin{description}
  \item[$\mathbf{2^{\prime}}.$] Let $J$ be the set of positive 
    integers $i$ such that $S_i$ intersects the support of the
    coarse-grained projectors $Q_1,Q_2,\ldots, Q_m$ centered around
    the cut. Then $L$ is an upper bound on the number of qudits in
    any subset $S_i$ with $i\in J$, and an upper bound on the number
    of projector terms in $H_i$ whenever $i\in J$.
\end{description}

To establish \Thm{thm:subr_cut} we will show that conditions
$\mathbf{1.,2'.,3.}$ can be satisfied by a decomposition of the 2D
grid into concentric bands as shown in \Fig{fig:transformation} (c).

\begin{theorem}[\textbf{Subvolume scaling for a rectangular cut}]
\label{thm:subr_cut2}
  Let $\ket{\Omega}$ be the unique ground state of a
  frustration-free Hamiltonian \Eq{eq:twoloc} on an $n\times L$ grid
  of qudits with local dimension $d$. Its entanglement entropy
  across a rectangular cut with the inner region $R\EqDef \{a+1,
  \ldots a+A \}\times \{b+1, b+2, \ldots b+B\}$ is at most 
  \begin{align}
    \frac{10^{17}(4A+4B)^{5/3}}{\gamma^{5/6}} 
      \log^{7/3}(4d(A+B)\gamma^{-1}).
  \label{eq:rectangle}
  \end{align}
\end{theorem}
\begin{proof}[Proof outline]

Without loss of generality, assume that $A>B$. For convenience, we
shall consider a larger rectangular 2D grid obtained by adding
ancilla qudits to ensure the following (see
\Fig{fig:transformation}):
\begin{itemize}
  \item The lattice is a rectangle of dimensions 
    $(A+2n')\times (B+2n')$, for some large positive integer $n'$.
    As before (see the remark after \Eq{def:Kmtl}) we will need the
    system size $n'$ to be sufficiently large in order for our AGSP
    of interest to be well defined.
    
  \item The region $R$ is centered with respect to the lattice. 
\end{itemize}
We add local terms to the Hamiltonian $H$ for each new plaquette, in
such a way that (a) the local spectral gap $\gamma$ is unchanged and
(b) the new Hamiltonian has a unique ground state
$\ket{\Omega}\otimes \ket{0}^{\otimes N_\mathrm{anc}}$ where
$N_{\mathrm{anc}}$ is the number of ancillary qudits added to the
grid. Note that the entanglement entropy of the ground state across
the given cut is therefore unchanged.  The new terms added to the
Hamiltonian are as follows: for each plaquette with $q<4$ old qudits
from original $n\times L$ grid and $4-q$ new ancilla qudits, we add
the projector $\id^{\otimes q}\otimes (\id - \ketbra{0}^{\otimes
4-q})$.

Now, as shown in \Fig{fig:transformation} (c), we group the
vertices of the lattice into concentric bands. Let the bands be
indexed by positive integers in increasing order, from smallest to
largest. The smallest band is the yellow rectangle in Figure
\ref{fig:transformation} (c), of dimensions $(A- B+1)\times 1$ if
$B$ is odd and $(A-B+2)\times 2$ if $B$ is even. We may then write
the Hamiltonian $H$ as 
\begin{align}
  H=\sum_{i}H'_{i}
  \label{eq:bandH}
\end{align}
where $H'_i$ contains all terms of $H$ which acts nontrivially
between the $i$th and $i+1$th band. Viewing \Eq{eq:bandH} as a 1D
chain of bands, we are interested in the entanglement entropy of the
ground state across the given cut separating the $c$th and $c+1$th
band, where
\begin{align*}
  c\EqDef \left\lceil\frac{B-2}{2}\right\rceil
\end{align*}
The decomposition \Eq{eq:bandH} therefore satisfies conditions
$\mathbf{1., 3.}$ defined above with respect to the partition of the
grid into bands. However, note that the number of qudits in the
$i$th band and the number of local projectors in each term $H'_i$
increases with the index $i$. Previously, for the 1D chain of
columns described by \Eq{eq:twoloc}, each column consisted of $L$
qudits and each local term $H_i$ contained at most $L$ local
projectors. Now, from Figure \ref{fig:transformation} (c), it can be
noted that the number of qudits in the $i+1$th band is $8$ more than
the number in the $i$th band. Similarly, the number of projectors in
$H'_{i+1}$ is at most $8$ more than the number of projectors in
$H'_{i}$. Write
\begin{align*}
  L_0=2(A+B)
\end{align*}
for the number of projectors in the term $H'_c$ which crosses the
cut of interest, and note that the number of qudits in the $c$th
band is $L_0-4$.

Now consider an operator $K(m,t,\ell)$ and choices of $m,t,\ell$
defined exactly the same way as in Sections \ref{sec:agsp} and
\ref{sec:arealaw1D}, but with the replacements $H_i \rightarrow
H'_i$, ``column" $\rightarrow$ ``band", and $L\rightarrow
L'=4(A+B)$. Note that with these replacements, the coarse grained
projectors $Q_j$ for $j=1,2,\ldots, m$ have support contained in a
contiguous region around the cut consisting of the bands
\begin{align}
  i\in \{c- 3mt+t+1,\ldots, c,c+1,\ldots, \ldots c+ 3mt-t\}.
\label{eq:interval}
\end{align}
The number of qudits in each of the bands \Eq{eq:interval} is at
most 
\begin{align}
  L_0-4+8(3mt-t)\leq 2(A+B)+24mt
\label{eq:bandupper}
\end{align}
and the number of local Hamiltonian terms in $H_i$ for $i$ in the
set \Eq{eq:interval} is also upper bounded by the right-hand side of
\Eq{eq:bandupper}. As long as $2(A+B)+24mt = L'/2+24mt$ is at most $L'$ (given the prescribed
choices of $m$ and $t$), the Hamiltonian \Eq{eq:bandH} satisfies
conditions $\mathbf{1.,2'.,3.}$ with $L\rightarrow L'$ and the proof
of \Thm{thm:subv_cut} goes through exactly the same as before. In
this case we obtain the bound \Eq{eq:rectangle} on the entanglement
entropy which is just the right-hand side of \Eq{eq:lastinequality}
with the replacement $L\rightarrow L'$.  If, on the other hand, we
find that the prescribed choices of $m$ and $t$ lead to the opposite
inequality
\begin{align*}
  L'/2<24mt,
\end{align*}
then, substituting Eqs.~(\ref{eq:tupperbound},~\ref{eq:mbounds})
with $L\rightarrow L'$, we get
\begin{align*}
  (L')^2< 48 mtL'\leq 48\cdot 26m^2L'\gamma^{-1/2}
    \leq \frac{5\cdot 10^{11}L^{5/3}}{\gamma^{5/6}} 
      \log^{4/3}(dL\gamma^{-1})
\end{align*}
In this case the trivial volume bound upper bounds the entanglement
entropy as
\begin{align*}
  (2(A+B))^2\log(d)\leq L'^2 \log(d)
    \leq \frac{5\cdot 10^{11}L^{5/3}}{\gamma^{5/6}} 
      \log^{7/3}(dL\gamma^{-1}),
\end{align*}
completing the proof.

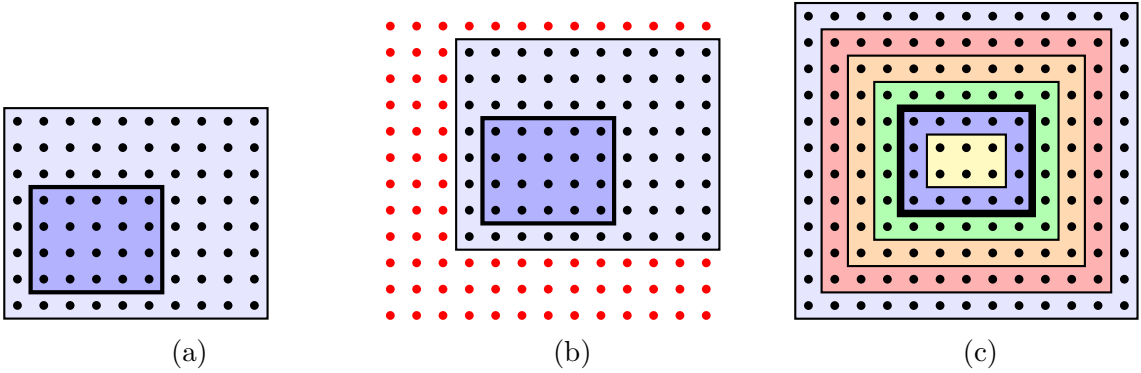
\begin{figure}
\centering
\begin{subfigure}[b]{0.3\textwidth}

\begin{tikzpicture}[xscale=0.35,yscale=0.35]


\draw [fill=blue!10!white, thick] (0.5,0.5+2) rectangle (10.5, 8.5+2);

\draw [fill=blue!30!white, ultra thick] (1.5,1.5+2) rectangle (6.5,5.5+2);

\foreach \i in {1,...,10}
{
\foreach \j in {1,...,8}
   \draw (\i, \j+2) node[circle, fill=black, scale=0.3]{};
}
\end{tikzpicture}
\caption{}
\end{subfigure}
\begin{subfigure}[b]{0.3\textwidth}
\begin{tikzpicture}[xscale=0.35,yscale=0.35]

\draw [fill=blue!10!white, thick] (0.5,0.5-9) rectangle (10.5, 8.5-9);

\draw [fill=blue!30!white, ultra thick] (1.5,1.5-9) rectangle (6.5,5.5-9);

\foreach \i in {1,...,10}
{
\foreach \j in {1,...,8}
   \draw (\i, \j-9) node[circle, fill=black, scale=0.3]{};
}

\foreach \i in {-2,...,0}
{
\foreach \j in {-2,...,8}
   \draw (\i, \j-9) node[circle, fill=red, scale=0.3]{};
}

\foreach \i in {1,...,10}
{
\foreach \j in {-2,...,0}
   \draw (\i, \j-9) node[circle, fill=red, scale=0.3]{};
}

\foreach \i in {-2,...,10}
   \draw (\i, 0) node[circle, fill=red, scale=0.3]{};
\end{tikzpicture}
\caption{}
\end{subfigure}
\hspace{0.2cm}
\begin{subfigure}[b]{0.3\textwidth}

\begin{tikzpicture}[xscale=0.35,yscale=0.35]

\draw [thick, fill=blue!10!white] (0.5, 0.5) rectangle (13.5,12.5);

\draw [thick, fill=red!30!white] (1.5,1.5) rectangle (12.5,11.5);

\draw [thick, fill=orange!30!white] (2.5,2.5) rectangle (11.5,10.5);

\draw [thick, fill=green!30!white] (3.5,3.5) rectangle (10.5,9.5);

\draw [line width=1mm, fill=blue!30!white] (4.5,4.5) rectangle (9.5,8.5);

\draw [thick, fill=yellow!30!white] (5.5,5.5) rectangle (8.5,7.5);

\foreach \i in {1,...,13}
{
\foreach \j in {1,...,12}
   \draw (\i, \j) node[circle, fill=black, scale=0.3]{};
}
\end{tikzpicture}
\caption{}
\end{subfigure}
  \caption{(a) \small The original lattice and region $R$. (b) Added
  ancilla qudits (in red) are used to transform the lattice into one
  of similar shape, such that $R$ is centered. (c) The lattice can
  be divided into a family of concentric rectangular bands. The cut
  bipartitioning the lattice into the region $R$ and its complement
  is shown in bold.  \label{fig:transformation}}
\end{figure}

\end{proof}

\section{Acknowledgments} 

IA acknowledges the support of the Israel Science Foundation (ISF) under the Individual Research Grant No.~1778/17. DG thanks Justin Thaler for helpful discussions about polynomials. DG acknowledges the support of the Natural Sciences and Engineering Research Council of Canada (NSERC) under Discovery grant number RGPIN-2019-04198.  IQC and PI are supported in part by the Government of Canada and the Province of Ontario.

\bibliographystyle{abbrv} 
\bibliography{references1}


\appendix

\section{Robust AND polynomial}
\label{append:ANDpoly}

We now provide a proof of \Thm{corrobust}, following \Ref{Sherstov12}.
\begin{proof}[Proof of \Thm{corrobust}]
For $t\in \mathbbm{R}\setminus{\{0\}}$ let $\sign(t)=t/|t|$ denote the sign of $t$. We may equivalently write
\begin{align*}
\sign(t)=\frac{t}{\sqrt{1+(t^2-1)}},
\end{align*}
and we may then use the binomial series to expand the denominator (see, e.g., Eq. 3.2 of \cite{Sherstov12}). This gives the following series expansion which converges for $0<|t|< \sqrt{2}$ 
\begin{align}
\sign(t)=t\sum_{i=0}^{\infty} \left(-\frac{1}{4}\right)^i {2i\choose i}\br{t^2-1}^i \qquad \quad 0<|t|< \sqrt{2}.
\label{eq:signfunc}
\end{align}

Now consider the following robust function for the Boolean monomial:
$$\inte(x) = \frac{1+\sign(2x-1)}{2} = \begin{cases} 1 & \text{ if }  x>\frac{1}{2}\\  0 & \text{ if } x< \frac{1}{2}  \end{cases}$$
Define
\begin{align*}
S=\left\{ x\in \mathbbm{R}: 0<|2x-1| < \sqrt{2}\right\}.
\end{align*}
For $x\in S$ we may use \Eq{eq:signfunc} and separate out the $i=0$ term to express $\inte(x)$ as 
\begin{align*}
\inte(x) =x+\frac{2x-1}{2}\sum_{i=1}^{\infty} {2i\choose i}\br{x(1-x)}^i=\sum_{i=0}^{\infty} A_i(x)
\end{align*}
where we define polynomials
$$A_0(x)\EqDef x \quad \text{and} \quad A_{i}(x) \EqDef \frac{2x-1}{2}{2i\choose i}\br{x(1-x)}^i \text{ for } i\geq 1.$$  Observe that $A_i$ has real coefficients and degree $2i+1$, for all $i\geq 0$. For $(x_1, x_2, \ldots x_m) \in S^m$,
\begin{eqnarray}
\label{eq:prodint}
\inte(x_1)\cdot \inte(x_2)\ldots \inte(x_m) &=& \sum_{i_1, i_2, \ldots i_m} A_{i_1}(x_1)A_{i_2}(x_2)\ldots A_{i_m}(x_m)\nonumber\\
&=&\sum_{n=0}^{\infty}\sum_{i_1, i_2, \ldots i_m: i_1+\ldots+ i_m=n} A_{i_1}(x_1)A_{i_2}(x_2)\ldots A_{i_m}(x_m)\nonumber\\
&\EqDef& \sum_{n=0}^{\infty} \xi_{n}(x_1, \ldots x_m).
\end{eqnarray}
Below we shall establish the following claim:
\begin{claim}
\label{clm:xierror}
For $(x_1, x_2,\ldots x_m)\in \br{\Br{-\frac{1}{20}, \frac{1}{20}} \cup [1-1/20,1+1/20]}^m$,
\begin{equation*}
|\xi_n(x_1, x_2, \ldots x_m)| \leq 3^m \br{\frac{3}{5}}^n.
\end{equation*}
\end{claim}
Let us define the robust polynomial $p_{AND}$ by truncating the sum in Equation \ref{eq:prodint} to $n\leq 5m$:
\begin{align}
\label{eq:panddef}
p_{AND}(x_1, \ldots x_m) \EqDef \sum_{n=0}^{5m} \xi_{n}(x_1, \ldots x_m).
\end{align}
 Since each $A_i$ is a univariate polynomial with real coefficients and degree $2i+1$, $p_{AND}$ has real coefficients and degree
\begin{align}
\label{eq:degpAND}
\max_{i_1+i_2+\ldots i_m \leq 5m} \br{(2i_1+1)+(2i_2+1)+\ldots (2i_m+1)} = 2(5m) + m = 11m.
\end{align}
In addition,
\begin{align}
\label{eq:pAND1111}
p_{AND}(1,1,\ldots 1) = \sum_{n=0}^{5m} \xi_{n}(1,1 \ldots 1) = \xi_0(1,1,\ldots 1) =1,
\end{align}
where we used the identity $A_i(1)=0$ for $i\geq 1$ and $A_0(1)=1$. Finally, suppose $x=(x_1,x_2,\ldots, x_m)=y+\epsilon$ where $y\in \{0,1\}^m$ and $\epsilon\in [-1/20, 1/20]^m$. Then for each $1\leq i\leq m$ we have
\begin{align*}
x_i\in  [-1/20,1/20]\cup [1-1/20, 1+1/20] \subset S
\end{align*}
and 
\begin{align*}
|p_{AND}(y+\epsilon) - y_1y_2\ldots y_m|=|p_{AND}(x_1,x_2,\ldots, x_m)-\inte(x_1)\inte(x_2)\ldots \inte(x_m)|.
\end{align*}
Using Eqs.~(\ref{eq:prodint}, \ref{eq:panddef}) and the triangle inequality to bound the right-hand side gives
\begin{align*}
|p_{AND}(y+\epsilon) - y_1y_2\ldots y_m|&\leq \sum_{n=5m+1}^{\infty} |\xi_{n}(x_1, \ldots x_m)|\\
&\leq 3^m\sum_{n=5m+1}^{\infty} \br{\frac{3}{5}}^n\\
&=3^m \left(\frac{3}{5}\right)^{5m} \cdot \frac{3}{2}\\
&\leq \left(3\cdot (3/5)^5\cdot (3/2)\right)^m.
\end{align*}
Noting that $3\cdot (3/5)^5\cdot (3/2)\leq e^{-1}$ we arrive at \Eq{eq:robusteq} and complete the proof.
\end{proof}

\begin{proof}[Proof of Claim \ref{clm:xierror}]
Define $J=[-1/20,1/20]\cup[1-1/20, 1+1/20]$ and note that for all $i\geq 1$ we have
\begin{align}
\max_{x\in J}|A_i(x)| = {2i\choose i}\max_{x\in J}\left|\frac{2x-1}{2}\br{x(1-x)}^i\right|  \leq 4^i\cdot \left(\frac{1}{20}\right)^i\left(\frac{21}{20}\right)^i \leq \br{\frac{21}{100}}^i,
\label{eq:i1}
\end{align}
where we used the fact that ${2i\choose i}\leq 4^i$, $\max_{x\in J} |\frac{2x-1}{2}|\leq 1$, and $\max_{x\in J}|x(1-x)|\leq (1/20)(21/20)$. Furthermore,
\begin{align}
\max_{x\in J}|A_0(x)| = \max_{x\in J}|x| \leq \frac{21}{20}.
\label{eq:i0}
\end{align}
Combining Eqs~(\ref{eq:i1}, \ref{eq:i0}) we see that for all $i\geq 0$,
\begin{align}
\label{eq:Aerror}
\max_{x\in J}|A_i(x)| \leq \left(\frac{21}{20}\right)\br{\frac{21}{100}}^i.
\end{align}
Consequently, for $(x_1, \ldots x_m)\in J^m$, using the definition of $\xi_n$ and the triangle inequality, we get
\begin{align}
|\xi_n(x_1, \ldots x_m)|&\leq \sum_{i_1, i_2, \ldots i_m: i_1+\ldots i_m=n} \left|A_{i_1}(x_1)A_{i_2}(x_2)\ldots A_{i_m}(x_m)\right|\nonumber\\
&\leq \left(\frac{21}{20}\right)^m \sum_{i_1, i_2, \ldots i_m: i_1+\ldots i_m=n} \br{\frac{21}{100}}^{i_1+i_2+\ldots i_m}\\
&= \left(\frac{21}{20}\right)^m \left(\frac{21}{100}\right)^n {m+n-1 \choose n-1}.\label{eq:almostdone}
\end{align}
where we used \Eq{eq:Aerror} and the fact that the number of tuples $(i_1,i_2,\ldots, i_m)$ of nonnegative integers satisfying $i_1+i_2+\ldots +i_m=n$ is given by ${m+n-1 \choose n-1}$. Finally, we substitute the bound ${m+n-1 \choose n-1}\leq 2^{m+n}$ into \Eq{eq:almostdone} to arrive at
\begin{align*}
|\xi_n(x_1, \ldots x_m)|\leq \left(\frac{42}{20}\right)^m \left(\frac{42}{100}\right)^n\leq 3^m (3/5)^n.
\end{align*}
\end{proof}

\section{Proof of Lemma \ref{lem:DL}}
\label{append:DLproof}
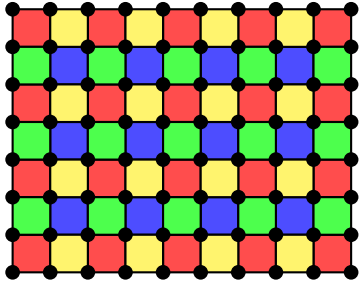
\begin{figure}
\centering
\begin{tikzpicture}[xscale=0.5,yscale=0.5]


\foreach \i in {1,...,5}
{
\foreach \j in {1,...,4}
\draw [fill=white!30!red, thick] (2*\i-1,2*\j-1) rectangle (2*\i,2*\j);
}

\foreach \i in {1,...,4}
{
\foreach \j in {1,...,3}
\draw [fill=white!30!blue, thick] (2*\i,2*\j) rectangle (2*\i+1,2*\j+1);
}

\foreach \i in {1,...,5}
{
\foreach \j in {1,...,3}
\draw [fill=white!30!green, thick] (2*\i-1,2*\j) rectangle (2*\i,2*\j+1);
}

\foreach \i in {1,...,4}
{
\foreach \j in {1,...,4}
\draw [fill=white!30!yellow, thick] (2*\i,2*\j-1) rectangle (2*\i+1,2*\j);
}

\foreach \i in {1,...,10}
{
\foreach \j in {1,...,8}
   \draw (\i, \j) node[circle, fill=black, scale=0.5]{};
}

\end{tikzpicture}
  \caption{\small The local projectors can be divided into $4$
  groups, where the projectors in each group commute with each
  other. \label{fig:commutegroup}}
\end{figure}

The proof is similar to that given in \cite{AAV16}, which uses a Chebyshev polynomial function of the detectability operator, as suggested in \cite{GossetH15}.  The projectors $\{P_{ij}\}$ can be divided into $4$ groups as follows (see Figure \ref{fig:commutegroup}), with the property that the projectors in each group commute with each other: 
$$\cG_1\EqDef \{P_{ij}: i=\mathrm{odd}, j=\mathrm{odd}\}, \quad \cG_2\EqDef \{P_{ij}: i=\mathrm{even}, j=\mathrm{odd}\},$$ $$\cG_3\EqDef \{P_{ij}: i=\mathrm{odd}, j=\mathrm{even}\}, \quad \cG_4\EqDef \{P_{ij}: i=\mathrm{even}, j=\mathrm{even}\}.$$
We also define
\begin{eqnarray}
\label{eq:DLlayers}
DL_k&\EqDef& \prod_{P_{ij}\in \cG_k}\br{\id-P_{ij}}, \qquad \quad 1\leq k\leq 4,
\end{eqnarray}
and define $DL \EqDef DL_4\cdot DL_3\cdot DL_2\cdot DL_1$. From \cite[Corollary 3]{AAV16}, it holds that for any $\psi$ satisfying $\langle \psi|\Omega\rangle=0$, we have 
\begin{align}
\label{eq:dl2DSV}
\|DL\ket{\psi}\|^2 \leq \frac{1}{1+\frac{\gamma}{8^2}} = \frac{1}{1+\frac{\gamma}{64}}.
\end{align} 
Here we used the fact that, for every projector $P_{ij}$, at most $8$ projectors do not commute with it. Now, we have the following claim, which is proved towards the end.  It uses the `light cone' argument from \cite{AharonovALV08}.
\begin{claim}
\label{clm:lightcone}
Let $F$ be any univariate polynomial of degree at most $t/6$ satisfying $F(0)=1$. Then
\begin{align}
DL(t)=\br{Q'_{2t}\cdot Q'_{8t}\cdot Q'_{14t}\cdot\ldots}\cdot F\left(\id- DL^{\dagger}DL\right)\cdot\br{Q'_{5t}\cdot Q'_{11t}\cdot Q'_{17t}\cdot\ldots}.
\label{eq:absorbF}
\end{align}
\end{claim}
Before proving this claim, we show how it can be used to establish Lemma \ref{lem:DL}. We apply the Claim with
 $F=\mathrm{Step}_{\frac{t}{6}, \frac{\gamma}{64+\gamma}}$ where the right-hand side is the polynomial from Fact \ref{fact:polyD}. From this we see that for any $\psi \in G_{\perp}$
\begin{align*}
\|DL(t)\ket{\psi}\|^2 \leq \|\mathrm{Step}_{\frac{t}{6}, \frac{\gamma}{64+\gamma}}\br{\id- DL^{\dagger}DL}\ket{\psi'}\|^2
\end{align*}
where $\psi'\in G_{\perp}$ is the state $$\ket{\psi'}=\br{Q'_{5t}\cdot Q'_{11t}\cdot Q'_{17t}\cdot\ldots}\ket{\psi}/\|\br{Q'_{5t}\cdot Q'_{11t}\cdot Q'_{17t}\cdot\ldots}\ket{\psi}\|.$$ But Eq. \eqref{eq:dl2DSV} ensures that the eigenvalues of $DL^{\dagger}DL$ in $G_\perp$ are at most $\frac{1}{1+ \frac{\gamma}{64}} = 1- \frac{\gamma}{64+\gamma}$. Using Fact \ref{fact:polyD} and the fact that $\gamma\leq 1$, we get
\begin{align*}
\|DL(t)\ket{\psi}\| \leq 2e^{-\frac{t}{3}\sqrt{\frac{\gamma}{64+\gamma}}}\leq 2e^{-\frac{t}{3}\sqrt{\frac{\gamma}{65}}}\leq 2e^{-\frac{t\sqrt{\gamma}}{25}}
\end{align*}
\begin{proof}[Proof of Claim \ref{clm:lightcone}] For every $i\in [n-1]$ and $k\in [4]$, let $$\Pi_{i,k}\EqDef \prod_{\id-P_{ij}\in \cG_k: Supp(P_{ij})\in \{i,i+1\}}(\id-P_{ij})$$ be the product of projectors from $\cG_k$ that are supported only on columns $\{i,i+1\}$. Since all projectors in $\cG_k$ commute, $\Pi_{i,k}$ is also a projector and we can write $DL_k = \prod_{i\in [n-1]}\Pi_{i,k}$. For any $S\subset [n]$, define 
$$DL_k^S\EqDef \prod_{i:Supp(\Pi_{i,k})\cap S \neq \phi} \Pi_{i,k}$$
as the product of projectors $\Pi_{i,k}$ that have their support overlapping with $S$.

The argument below has been illustrated in Figure \ref{fig:absorb}. Let $S_0$ be the complement of the support of $\br{Q'_{2t}\cdot Q'_{8t}\cdot Q'_{14t}\cdot\ldots}$.  Observe, using frustration-freeness, that for any $\Pi_{i,k}$ whose support is contained in the support of $\br{Q'_{2t}\cdot Q'_{8t}\cdot Q'_{14t}\cdot\ldots}$, we have 
$$\br{Q'_{2t}\cdot Q'_{8t}\cdot Q'_{14t}\cdot\ldots}\Pi_{i,k} = \br{Q'_{2t}\cdot Q'_{8t}\cdot Q'_{14t}\cdot\ldots}.$$ 
This implies the following identity (c.f. Figure \ref{fig:absorb} (b)):
\begin{align}
\label{eq:firstabsorb}
\br{Q'_{2t}\cdot Q'_{8t}\cdot Q'_{14t}\cdot\ldots}DL_1 = \br{Q'_{2t}\cdot Q'_{8t}\cdot Q'_{14t}\cdot\ldots}DL_1^{S_0}.
\end{align}
For all integers $\alpha\geq 1$, recursively define $S_{\alpha}$ as the set of all columns
 at distance at most $1$ from $S_{\alpha-1}$. Clearly, we have the inclusion $S_0 \subset S_1 \subset S_2 \ldots$. Similar to Eq. \eqref{eq:firstabsorb}, we can `absorb' some of the projectors in $DL_1DL_2$ and obtain the identity:
\begin{align}
\label{eq:secondabsorb}
\br{Q'_{2t}\cdot Q'_{8t}\cdot Q'_{14t}\cdot\ldots}DL_1DL_2=\br{Q'_{2t}\cdot Q'_{8t}\cdot Q'_{14t}\cdot\ldots}DL_1^{S_0}DL_2^{S_1}.
\end{align}
Applying the same argument recursively, and using the fact that $$\br{DL^{\dagger}DL}^p = \br{DL_1\cdot DL_2\cdot DL_3\cdot DL_4\cdot DL_3\cdot DL_2}^p\cdot DL_1,$$ we conclude (c.f. Figure \ref{fig:absorb} (c))
\begin{align}
\label{eq:evenabsorb}
\br{Q'_{2t}\cdot Q'_{8t}\cdot Q'_{14t}\cdot\ldots}\br{DL^{\dagger}DL}^{p} = \br{Q'_{2t}\cdot Q'_{8t}\cdot Q'_{14t}\cdot\ldots}DL_1^{S_0}\cdot DL_2^{S_1}\cdot DL_3^{S_2}\ldots DL_1^{S_{6p}}.
\end{align}
If $6p\leq t$, the set $S_{6p}$ is contained in the support of $\br{Q'_{5t}\cdot Q'_{11t}\cdot Q'_{17t}\cdot\ldots}$. Furthermore, if $\Pi_{i,k}$ is in the support of $\br{Q'_{5t}\cdot Q'_{11t}\cdot Q'_{17t}\cdot\ldots}$, we have 
$$\Pi_{i,k}\br{Q'_{5t}\cdot Q'_{11t}\cdot Q'_{17t}\cdot\ldots} = \br{Q'_{5t}\cdot Q'_{11t}\cdot Q'_{17t}\cdot\ldots}.$$
Thus, all the projectors in $DL_1^{S_0}\cdot DL_2^{S_1}\cdot DL_3^{S_2}\ldots DL_1^{S_{6p}}$ can be `absorbed' in $\br{Q'_{5t}\cdot Q'_{11t}\cdot Q'_{17t}\cdot\ldots}$, which can be formalized as:
\begin{align}
\label{eq:oddabsorb}
DL_1^{S_0}\cdot DL_2^{S_1}\cdot DL_3^{S_2}\ldots DL_1^{S_{6p}}\br{Q'_{5t}\cdot Q'_{11t}\cdot Q'_{17t}\cdot\ldots}= \br{Q'_{5t}\cdot Q'_{11t}\cdot Q'_{17t}\cdot\ldots}.
\end{align}
Combining Eqs. \eqref{eq:evenabsorb} and \eqref{eq:oddabsorb}, we find that
\begin{align}
\br{Q'_{2t}\cdot Q'_{8t}\cdot Q'_{14t}\cdot\ldots}\br{DL^{\dagger}DL}^{p}\br{Q'_{5t}\cdot Q'_{11t}\cdot Q'_{17t}\cdot\ldots}=\br{Q'_{2t}\cdot Q'_{8t}\cdot Q'_{14t}\cdot\ldots}\br{Q'_{5t}\cdot Q'_{11t}\cdot Q'_{17t}\cdot\ldots}
\label{eq:powerabsorb}
\end{align}
for any $p\leq t/6$. Thus, any such power $(DL^{\dagger}DL)^p$ can be replaced by $1$ whenever it is sandwiched between the products of projectors in \Eq{eq:powerabsorb}. This implies that for a polynomial $F$ of degree at most $t/6$, we have
\begin{align}
\br{Q'_{2t}\cdot Q'_{8t}\cdot Q'_{14t}\cdot\ldots}F(I-DL^{\dagger}DL)&\br{Q'_{5t}\cdot Q'_{11t}\cdot Q'_{17t}\cdot\ldots}=\\&\br{Q'_{2t}\cdot Q'_{8t}\cdot Q'_{14t}\cdot\ldots}F(0)\br{Q'_{5t}\cdot Q'_{11t}\cdot Q'_{17t}\cdot\ldots},
\label{eq:powerabsorb}
\end{align}
and using the fact that $F(0)=1$ completes the proof.
\end{proof}
\begin{figure}
\begin{tikzpicture}[xscale=1,yscale=1]


\draw [ultra thick, fill=green!30!white] (0.5, -1) rectangle (1.8,-0.5);
\draw[ultra thick, fill=green!30!white] (4.8, -1) rectangle (11.9,-0.5);
\draw[ultra thick, fill=green!30!white] (14.9, -1) rectangle (16.4,-0.5);

\draw[ultra thick, fill=red!30!white] (0.5,2.2) rectangle (6.8,1.7);
\draw[ultra thick, fill=red!30!white] (9.8,2.2) rectangle (16.4,1.7);

\draw[thick, fill=black, <->] (9.8, 2.5)--(16.4, 2.5);
\node at (13.5, 3){$4t$};
\draw[thick, fill=black, <->] (9.8, 3)--(11.9, 3);
\node at (11, 3.5){$t$};

\foreach \k in {0,...,14}
{
\draw [fill=white!70!blue] (0.9+ \k, -0.2) rectangle (1.75+\k, 0.05);
}
\draw [fill=white!70!blue] (0.5, -0.2) rectangle (0.75, 0.05);
\draw [fill=white!70!blue] (15.9, -0.2) rectangle (16.4, 0.05);

\foreach \k in {1,...,15}
{
\draw [fill=white!70!blue] (0.4+ \k, -0.35) rectangle (1.25+\k, -0.1);
}
\draw [fill=white!70!blue] (0.5, -0.35) rectangle (1.25, -0.1);

\foreach \k in {0,...,14}
{
\draw [fill=white!50!blue] (0.9+ \k, 0.3) rectangle (1.75+\k, 0.55);
}
\draw [fill=white!50!blue] (0.5, 0.3) rectangle (0.75, 0.55);
\draw [fill=white!50!blue] (15.9, 0.3) rectangle (16.4, 0.55);

\foreach \k in {1,...,15}
{
\draw [fill=white!50!blue] (0.4+ \k, 0.15) rectangle (1.25+\k, 0.4);
}
\draw [fill=white!50!blue] (0.5, 0.15) rectangle (1.25, 0.4);

\foreach \k in {0,...,14}
{
\draw [fill=white!70!yellow] (0.9+ \k, 0.8) rectangle (1.75+\k, 1.05);
}
\draw [fill=white!70!yellow] (0.5, 0.8) rectangle (0.75, 1.05);
\draw [fill=white!70!yellow] (15.9, 0.8) rectangle (16.4, 1.05);

\foreach \k in {1,...,15}
{
\draw [fill=white!70!yellow] (0.4+ \k, 0.65) rectangle (1.25+\k, 0.9);
}
\draw [fill=white!70!yellow] (0.5, 0.65) rectangle (1.25, 0.9);

\foreach \k in {0,...,14}
{
\draw [fill=white!50!yellow] (0.9+ \k, 1.3) rectangle (1.75+\k, 1.55);
}
\draw [fill=white!50!yellow] (0.5, 1.3) rectangle (0.75, 1.55);
\draw [fill=white!50!yellow] (15.9, 1.3) rectangle (16.4, 1.55);

\foreach \k in {1,...,15}
{
\draw [fill=white!50!yellow] (0.4+ \k, 1.15) rectangle (1.25+\k, 1.40);
}
\draw [fill=white!50!yellow] (0.5, 1.15) rectangle (1.25, 1.40);

\node at (8.5, -1.5) {(a)};


\draw [ultra thick, fill=green!30!white] (0.5, -1-5) rectangle (1.8,-0.5-5);
\draw[ultra thick, fill=green!30!white] (4.8, -1-5) rectangle (11.9,-0.5-5);
\draw[ultra thick, fill=green!30!white] (14.9, -1-5) rectangle (16.4,-0.5-5);

\draw[ultra thick, fill=red!30!white] (0.5,2.2-5) rectangle (6.8,1.7-5);
\draw[ultra thick, fill=red!30!white] (9.8,2.2-5) rectangle (16.4,1.7-5);

\foreach \k in {0,...,14}
{
\draw [fill=white!70!blue] (0.9+ \k, -0.2-5) rectangle (1.75+\k, 0.05-5);
}
\draw [fill=white!70!blue] (0.5, -0.2-5) rectangle (0.75, 0.05-5);
\draw [fill=white!70!blue] (15.9, -0.2-5) rectangle (16.4, 0.05-5);

\foreach \k in {1,...,15}
{
\draw [fill=white!70!blue] (0.4+ \k, -0.35-5) rectangle (1.25+\k, -0.1-5);
}
\draw [fill=white!70!blue] (0.5, -0.35-5) rectangle (1.25, -0.1-5);

\foreach \k in {0,...,14}
{
\draw [fill=white!50!blue] (0.9+ \k, 0.3-5) rectangle (1.75+\k, 0.55-5);
}
\draw [fill=white!50!blue] (0.5, 0.3-5) rectangle (0.75, 0.55-5);
\draw [fill=white!50!blue] (15.9, 0.3-5) rectangle (16.4, 0.55-5);

\foreach \k in {1,...,15}
{
\draw [fill=white!50!blue] (0.4+ \k, 0.15-5) rectangle (1.25+\k, 0.4-5);
}
\draw [fill=white!50!blue] (0.5, 0.15-5) rectangle (1.25, 0.4-5);

\foreach \k in {0,...,14}
{
\draw [fill=white!70!yellow] (0.9+ \k, 0.8-5) rectangle (1.75+\k, 1.05-5);
}
\draw [fill=white!70!yellow] (0.5, 0.8-5) rectangle (0.75, 1.05-5);
\draw [fill=white!70!yellow] (15.9, 0.8-5) rectangle (16.4, 1.05-5);

\foreach \k in {1,...,15}
{
\draw [fill=white!70!yellow] (0.4+ \k, 0.65-5) rectangle (1.25+\k, 0.9-5);
}
\draw [fill=white!70!yellow] (0.5, 0.65-5) rectangle (1.25, 0.9-5);

\foreach \k in {6,...,8}
{
\draw [fill=white!50!yellow] (0.9+ \k, 1.3-5) rectangle (1.75+\k, 1.55-5);
}

\foreach \k in {6,...,9}
{
\draw [fill=white!50!yellow] (0.4+ \k, 1.15-5) rectangle (1.25+\k, 1.40-5);
}

\node at (8.5, -1.5-5) {(b)};


\draw [ultra thick, fill=green!30!white] (0.5, -1-10) rectangle (1.8,-0.5-10);
\draw[ultra thick, fill=green!30!white] (4.8, -1-10) rectangle (11.9,-0.5-10);
\draw[ultra thick, fill=green!30!white] (14.9, -1-10) rectangle (16.4,-0.5-10);

\draw[ultra thick, fill=red!30!white] (0.5,2.2-10) rectangle (6.8,1.7-10);
\draw[ultra thick, fill=red!30!white] (9.8,2.2-10) rectangle (16.4,1.7-10);

\foreach \k in {4,...,10}
{
\draw [fill=white!70!blue] (0.9+ \k, -0.2-10) rectangle (1.75+\k, 0.05-10);
}
\draw [fill=white!70!blue] (0.5, -0.2-10) rectangle (0.75, 0.05-10);
\draw [fill=white!70!blue] (0.9, -0.2-10) rectangle (1.75, 0.05-10);
\draw [fill=white!70!blue] (14.9, -0.2-10) rectangle (15.75, 0.05-10);
\draw [fill=white!70!blue] (15.9, -0.2-10) rectangle (16.4, 0.05-10);

\foreach \k in {5,...,10}
{
\draw [fill=white!70!blue] (0.4+ \k, -0.35-10) rectangle (1.25+\k, -0.1-10);
}
\draw [fill=white!70!blue] (0.5, -0.35-10) rectangle (1.25, -0.1-10);
\draw [fill=white!70!blue] (15.4, -0.35-10) rectangle (16.25, -0.1-10);

\foreach \k in {5,...,9}
{
\draw [fill=white!50!blue] (0.9+ \k, 0.3-10) rectangle (1.75+\k, 0.55-10);
}
\draw [fill=white!50!blue] (0.5, 0.3-10) rectangle (0.75, 0.55-10);
\draw [fill=white!50!blue] (15.9, 0.3-10) rectangle (16.4, 0.55-10);

\foreach \k in {5,...,10}
{
\draw [fill=white!50!blue] (0.4+ \k, 0.15-10) rectangle (1.25+\k, 0.4-10);
}
\draw [fill=white!50!blue] (0.5, 0.15-10) rectangle (1.25, 0.4-10);
\draw [fill=white!50!blue] (15.4, 0.15-10) rectangle (16.25, 0.4-10);

\foreach \k in {5,...,9}
{
\draw [fill=white!70!yellow] (0.9+ \k, 0.8-10) rectangle (1.75+\k, 1.05-10);
}
\draw [fill=white!70!yellow] (0.5, 0.8-10) rectangle (0.75, 1.05-10);
\draw [fill=white!70!yellow] (15.9, 0.8-10) rectangle (16.4, 1.05-10);

\foreach \k in {6,...,9}
{
\draw [fill=white!70!yellow] (0.4+ \k, 0.65-10) rectangle (1.25+\k, 0.9-10);
}

\foreach \k in {6,...,8}
{
\draw [fill=white!50!yellow] (0.9+ \k, 1.3-10) rectangle (1.75+\k, 1.55-10);
}

\foreach \k in {6,...,9}
{
\draw [fill=white!50!yellow] (0.4+ \k, 1.15-10) rectangle (1.25+\k, 1.40-10);
}

\node at (8.5, -1.5-10) {(c)};\end{tikzpicture}
  \caption{\small Graphical description of Eq. \eqref{eq:evenabsorb}. \textbf{(a)} The operators $DL_1, DL_2, DL_3, DL_4$ correspond to the dark yellow, light yellow, dark blue and light blue layers, respectively. Within each layer, all the projectors (small rectangles representing $\Pi_{i,k}$) mutually commute, although they need not have disjoint support. Two projectors from different layers may not commute if they have overlapping support. {\bf(b):} Some projectors in $DL_1$ are `absorbed' by the red coarse-grained layer. Resulting operator is $DL_1^{S_0}$ from Eqn.~\eqref{eq:firstabsorb}. {\bf (c):} The same process occurs for $4$ steps, with projectors from $DL_2, DL_3, DL_4$ absorbed in the red coarse-grained layer. The resulting operator is $DL_1^{S_0}DL_2^{S_1}DL_3^{S_2}DL_4^{S_3}$ and the support of `unabsorbed' projectors increases its boundary by one at each step. All the remaining projectors can be absorbed in the green coarse-grained layer, as they are contained in its support. \label{fig:absorb}}
\end{figure}
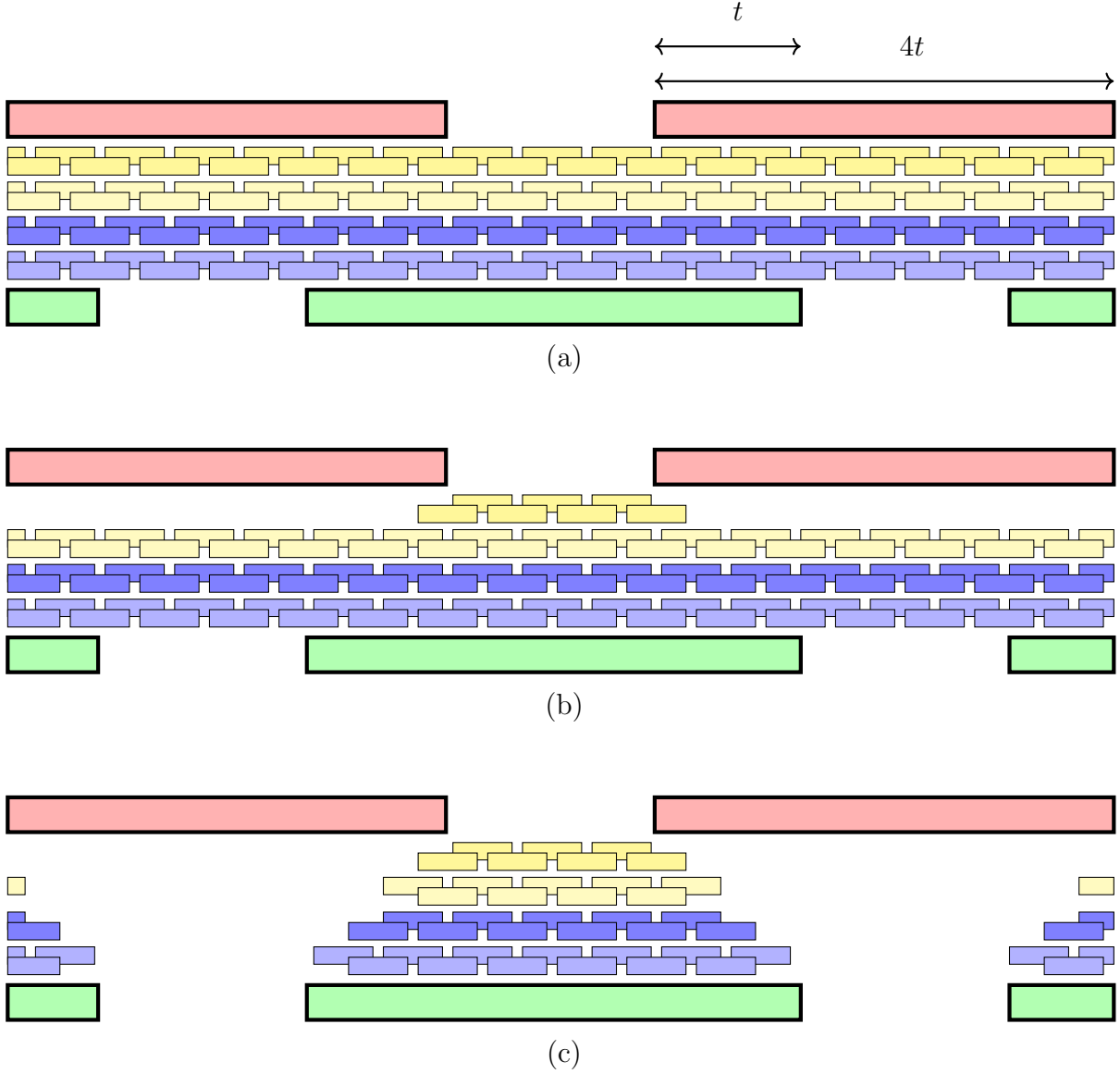

\section{R{\'e}nyi entanglement entropy and PEPS description}
\label{app:huangbound}
The R{\'e}nyi entropy of order $\alpha\in (0,1)$ is defined as 
$$S_{\alpha}(\rho) = \frac{1}{1-\alpha}\log\Tr(\rho^{\alpha}).$$
Below we show that the R{\'e}nyi entropy of order $1/2$ satisfies the same subvolume law as in \Thm{thm:subr_cut}. This implies, via an argument from \Ref{Huang19}, the following PEPS description of the ground state. In the following we say an operator $O$ is geometrically local if its support is contained in a contiguous region of $O(1)$ qudits.
\begin{theorem}
Let $\delta \in (0,1)$ and suppose $\gamma, d = \Omega(1)$. There exists a PEPS state $\ket{\phi}$ with bond dimension $e^{O\br{\frac{1}{\delta^5}\log^{21}\br{\frac{1}{\delta}}}}$ such that 
$$|\bra{\Omega}O\ket{\Omega} - \bra{\phi}O\ket{\phi}|\leq \delta,$$ for all geometrically local operators $O$ satisfying $\|O\|\leq 1$.
\end{theorem}
\begin{proof}[Proof sketch]
We essentially follow the proof of \cite[Theorem 3]{Huang19}, with a minor modification arising due to the fact that we are not considering periodic boundary conditions. Consider a partition of the lattice into regions shown in Figure \ref{fig:periodicsquare}. From Claim \ref{clm:renyi}, we conclude that $S_{\frac{1}{2}}$ across any of the blue rectangular bi-partitions (with perimeter $4m$) is upper bounded by $O(m^{5/3}\log^{7/3}(m))$. The width $b$ of green rectangles is chosen so that that the reduced ground state in any blue region can be purified (up to an error of $\frac{\delta}{2}$) in the associated green region. Using the upper bound on $S_{\frac{1}{2}}$ from Claim \ref{clm:renyi}, the choice $b= O(m^{2/3}\log^{7/3}(m)\log\frac{1}{\delta})$ suffices, see Eq. (36) of Ref.~\cite{Huang19}. The PEPS state is constructed in the same manner as given in \cite[Theorem 3]{Huang19}. The error in approximating the local expectation value arises in two ways: first error of $\frac{\delta}{2}$ incurred in approximate purification of the reduced density matrix and the second error of $O\br{\frac{b}{m}}$ incurred if a local operator overlaps a green region for some choice of $a,a'$ (see Figure \ref{fig:expandedlat}). This gives a total error of (cf. Eq. (38) of Ref.~\cite{Huang19}) $$\frac{\delta}{2}+ O\br{\frac{b}{m}}=\frac{\delta}{2} + O\br{\frac{m^{5/3}\log^{7/3}(m)\log\frac{1}{\delta}}{m^2}} \leq \delta,$$ if we choose $m= \Omega\br{\frac{1}{\delta^3}\log^{11}\br{\frac{1}{\delta}}}$. This implies the existence of a PEPS state $\ket{\phi}$ with bond dimension $e^{O\br{m^{5/3}\log^{7/3}(m)}} = e^{O\br{\frac{1}{\delta^5}\log^{21}\br{\frac{1}{\delta}}}}$, which completes the proof.

\begin{figure}
\centering
\begin{tikzpicture}[xscale=0.5,yscale=0.5]

\draw [fill=blue!10!white, thick] (0.4,1.4) rectangle (14.6, 11.6);

\foreach \j in {1,...,2}
\draw [fill=blue!30!white, thick] (0.7,0.6+3*\j) rectangle (2.3,3.3+3*\j);
\draw [fill=blue!30!white, thick] (0.7,0.6+1) rectangle (2.3,2.3+1);
\draw [fill=blue!30!white, thick] (0.7,0.6+9) rectangle (2.3,2.3+9);

\foreach \j in {1,...,2}
\draw [fill=green!30!white, thick] (2.7,0.6+3*\j) rectangle (3.3,3.3+3*\j);
\draw [fill=green!30!white, thick] (2.7,0.6+1) rectangle (3.3,2.3+1);
\draw [fill=green!30!white, thick] (2.7,0.6+9) rectangle (3.3,2.3+9);

\foreach \j in {1,...,2}
\draw [fill=blue!30!white, thick] (3.7,0.6+3*\j) rectangle (6.3,3.3+3*\j);
\draw [fill=blue!30!white, thick] (3.7,0.6+1) rectangle (6.3,2.3+1);
\draw [fill=blue!30!white, thick] (3.7,0.6+9) rectangle (6.3,2.3+9);

\foreach \j in {1,...,2}
\draw [fill=green!30!white, thick] (6.7,0.6+3*\j) rectangle (7.3,3.3+3*\j);
\draw [fill=green!30!white, thick] (6.7,0.6+1) rectangle (7.3,2.3+1);
\draw [fill=green!30!white, thick] (6.7,0.6+9) rectangle (7.3,2.3+9);

\foreach \j in {1,...,2}
\draw [fill=blue!30!white, thick] (7.7,0.6+3*\j) rectangle (10.3,3.3+3*\j);
\draw [fill=blue!30!white, thick] (7.7,0.6+1) rectangle (10.3,2.3+1);
\draw [fill=blue!30!white, thick] (7.7,0.6+9) rectangle (10.3,2.3+9);

\foreach \j in {1,...,2}
\draw [fill=green!30!white, thick] (10.7,0.6+3*\j) rectangle (11.3,3.3+3*\j);
\draw [fill=green!30!white, thick] (10.7,0.6+1) rectangle (11.3,2.3+1);
\draw [fill=green!30!white, thick] (10.7,0.6+9) rectangle (11.3,2.3+9);

\foreach \j in {1,...,2}
\draw [fill=green!30!white, thick] (11.7,0.6+3*\j) rectangle (12.3,3.3+3*\j);
\draw [fill=green!30!white, thick] (11.7,0.6+1) rectangle (12.3,2.3+1);
\draw [fill=green!30!white, thick] (11.7,0.6+9) rectangle (12.3,2.3+9);

\foreach \j in {1,...,2}
\draw [fill=blue!30!white, thick] (12.7,0.6+3*\j) rectangle (14.3,3.3+3*\j);
\draw [fill=blue!30!white, thick] (12.7,0.6+1) rectangle (14.3,2.3+1);
\draw [fill=blue!30!white, thick] (12.7,0.6+9) rectangle (14.3,2.3+9);

\foreach \i in {1,...,14}
{
\foreach \j in {2,...,11}
   \draw (\i, \j) node[circle, fill=black, scale=0.3]{};
}

\draw[thick, |-|] (0, 9.8)--(0,11.2);
\node at (-0.5, 10.5) {$a'$};

\draw[thick, |-|] (0, 6.8)--(0,9.2);
\node at (-0.5, 8) {$m$};

\draw[thick, |-|] (0.7, 12)--(2.2,12);
\node at (1.5, 12.5) {$a$};

\draw[thick, |-|] (3.7, 12)--(6.2,12);
\node at (5, 12.5) {$m$};

\draw[thick, |-|] (6.7, 12)--(7.3,12);
\node at (7, 12.5) {$b$};

\end{tikzpicture}
  \caption{\small A partition of the lattice into several blue and green parts. The size of green parts is chosen so that the reduced ground state on blue regions (except the rightmost ones) can be purified in the green regions to their immediate right. The reduced ground state on the rightmost blue regions can be purified in the green regions to their immediate left. The subvolume law in \Thm{thm:subr_cut} ensures that the size of green regions can be chosen much smaller than that of the blue regions. If $a, a'$ are chosen uniformly and independently at random in the range $\{1,2,\ldots m\}$, the probability that a local operator is not supported in a blue region is $O\br{\frac{b}{m}}$. This is depicted in Figure \ref{fig:expandedlat}. \label{fig:periodicsquare}}
\end{figure}
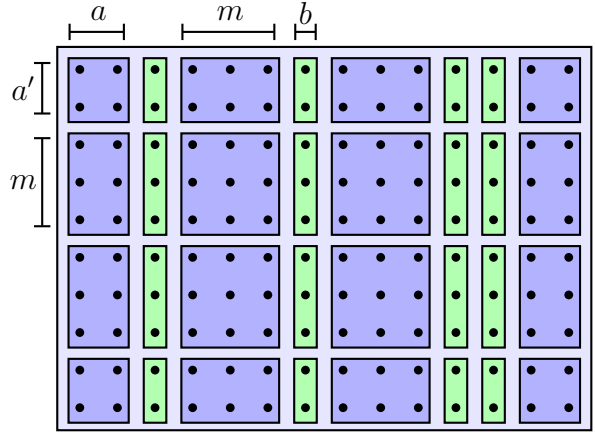
\end{proof}

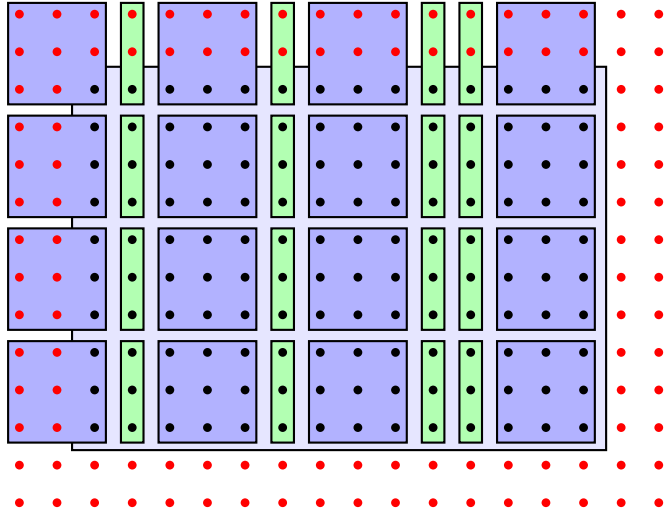
\begin{figure}
\centering
\begin{tikzpicture}[xscale=0.5,yscale=0.5]

\draw [fill=blue!10!white, thick] (0.4,1.4) rectangle (14.6, 11.6);

\foreach \j in {0,...,3}
\draw [fill=blue!30!white, thick] (-1.3,1.6+3*\j) rectangle (1.3,4.3+3*\j);

\foreach \j in {0,...,3}
\draw [fill=green!30!white, thick] (1.7,1.6+3*\j) rectangle (2.3,4.3+3*\j);

\foreach \j in {0,...,3}
\draw [fill=blue!30!white, thick] (2.7,1.6+3*\j) rectangle (5.3,4.3+3*\j);

\foreach \j in {0,...,3}
\draw [fill=green!30!white, thick] (5.7,1.6+3*\j) rectangle (6.3,4.3+3*\j);

\foreach \j in {0,...,3}
\draw [fill=blue!30!white, thick] (6.7,1.6+3*\j) rectangle (9.3,4.3+3*\j);

\foreach \j in {0,...,3}
\draw [fill=green!30!white, thick] (9.7,1.6+3*\j) rectangle (10.3,4.3+3*\j);

\foreach \j in {0,...,3}
\draw [fill=green!30!white, thick] (10.7,1.6+3*\j) rectangle (11.3,4.3+3*\j);

\foreach \j in {0,...,3}
\draw [fill=blue!30!white, thick] (11.7,1.6+3*\j) rectangle (14.3,4.3+3*\j);

\foreach \i in {1,...,14}
{
\foreach \j in {2,...,11}
   \draw (\i, \j) node[circle, fill=black, scale=0.3]{};
}

\foreach \i in {-1,...,16}
{
   \draw (\i, 0) node[circle, fill=red, scale=0.3]{};
   \draw (\i, 1) node[circle, fill=red, scale=0.3]{};
 \draw (\i, 12) node[circle, fill=red, scale=0.3]{};
 \draw (\i, 13) node[circle, fill=red, scale=0.3]{};

}

\foreach \j in {2,...,11}
{  
\draw (-1, \j) node[circle, fill=red, scale=0.3]{};
   \draw (0, \j) node[circle, fill=red, scale=0.3]{};
   \draw (15, \j) node[circle, fill=red, scale=0.3]{};
   \draw (16, \j) node[circle, fill=red, scale=0.3]{};
}

\end{tikzpicture}
  \caption{\small Expand the lattice by adding $m-1$ rows and
  columns of red vertices in each direction. Consider the partition
  of new lattice, as shown above. Translating this partition by
  $a-1$ steps right and $a'-1$ steps down gives the partition of the
  original lattice in Figure \ref{fig:periodicsquare}. Now consider
  any local operator within the original lattice. Translating the
  original partition by $a-1$ steps right and $a'-1$ steps down is
  equivalent to translating the operator by $a-1$ steps left and
  $a'-1$ steps up. If $a,a'$ are chosen at random in $\{1,2,\ldots
  m\}$, the probability that the operator is not contained in a blue
  region can now easily be computed to be
  $O\br{\frac{b}{m}}$.  \label{fig:expandedlat}}
\end{figure}

\begin{claim}
\label{clm:renyi}
Under the conditions of \Thm{thm:subr_cut}, for a bipartition corresponding to a rectangular region $A$ and its complement, we have
\begin{align*}
S_{\frac{1}{2}}(\Omega_A)\leq O\left(\frac{|\partial A|^{5/3}}{\gamma^{5/6}} 
      \log^{7/3}(d|\partial A|\gamma^{-1})\right).
\end{align*}
\end{claim}
\begin{proof}
The proof of \Thm{thm:subr_cut} establishes the existence of a $(D,\Delta)$-AGSP with respect to the given rectangular bipartition such that $D\Delta<1/2$ and (cf. \Thm{thm:AGSParealaw})
\begin{align*}
10\log(D)\leq \frac{C|\partial A|^{5/3}}{\gamma^{5/6}} 
      \log^{7/3}(d|\partial A|\gamma^{-1})
\end{align*}
for some universal constant $C>0$. To complete the proof, we show $S_{1/2}(\Omega_A)\leq 10\log(D)$ for any $(D,\Delta)$-AGSP satisfying $D\Delta<1/2$. 

As shown in \cite[Proof of Lemma 3.3]{AradLV12}, a $(D,\Delta)$ AGSP with $D\Delta<\frac{1}{2}$ implies the following bound on the Schmidt coefficients $\{\lambda_1, \lambda_2, \ldots \}$ (arranged in non-increasing order) of $\ket{\Omega}$, with respect to the cut:
$$\sum_{i\in \{D^{\ell}+1, \ldots D^{\ell+1}\}}\lambda_i^2\leq \sum_{i>D^{\ell}}\lambda_i^2 \leq 2D\Delta^{\ell} \leq \Delta^{\ell-1},$$
for all integers $\ell\geq 1$. We will upper bound $S_{\frac{1}{2}}(\Omega_R) = 2\log\br{\sum_{i}\lambda_i}$, under the above constraint. Following \cite{AradLV12}, we can maximize the R{\'e}nyi entropy by setting $\lambda^2_i = \frac{\Delta^{\ell-1}}{D^{\ell+1}- D^{\ell}}$, whenever $i\in \{D^{\ell}+1, \ldots D^{\ell+1}\}$. With this choice,
\begin{eqnarray*}
\sum_{i}\lambda_i &=&  \sum_{i\leq D}\lambda_i + \sum_{\ell=1}^{\infty}\br{(D^{\ell+1}- D^{\ell})\sqrt{\frac{\Delta^{\ell-1}}{D^{\ell+1}- D^{\ell}}}}\\
&\leq& D + \sum_{\ell=1}^{\infty}\br{\sqrt{D^{\ell+1}\Delta^{\ell-1}}} = D + D \sum_{\ell=1}^{\infty} \sqrt{(D\Delta)^{\ell-1}}\\
&\leq& D\br{1+ \sum_{\ell=1}^{\infty}\br{\frac{1}{\sqrt{2}}}^{\ell-1}} \leq 5D.
\end{eqnarray*}  
Thus, $S_{\frac{1}{2}}(\Omega_R) \leq 2\log(5D) \leq 10\log(D)$, completing the proof.
\end{proof}

\end{document}